\documentclass[10pt]{article}
\usepackage[T1]{fontenc}
\usepackage{amsmath,amsfonts,amsthm,mathrsfs,amssymb}
\usepackage{cite}
\usepackage{multirow}
\usepackage{graphicx,float}
\usepackage{subfigure}
\usepackage{placeins}
\usepackage{color}
\usepackage{indentfirst}
\usepackage[bookmarks=true]{hyperref}
\numberwithin{equation}{section}
\newcommand{\R}{{\mathbb R}}

\newcommand{\be}{\begin{eqnarray}}
\newcommand{\ben}{\begin{eqnarray*}}
\newcommand{\en}{\end{eqnarray}}
\newcommand{\enn}{\end{eqnarray*}}

\newcommand{\pa}{\partial}

\newtheorem{theorem}{Theorem}[section]
\newtheorem{lemma}[theorem]{Lemma}

\newtheorem{remark}[theorem]{Remark}

\definecolor{rot}{rgb}{1.000,0.000,0.000}
\begin{document}
\renewcommand{\theequation}{\arabic{section}.\arabic{equation}}
\begin{titlepage}
  \title{An accurate hyper-singular boundary integral equation method for dynamic poroelasticity in two dimensions}

\author{Lu Zhang\thanks{School of Mathematical Sciences, University of Electronic Science and Technology of China, Chengdu, Sichuan 611731, China. Email: {\tt zl129354@163.com}}\;,
Liwei Xu\thanks{School of Mathematical Sciences, University of Electronic Science and Technology of China, Chengdu, Sichuan 611731, China. Email: {\tt xul@uestc.edu.cn}}\;,
Tao Yin\thanks{Department of Computing \& Mathematical Sciences, California Institute of Technology, 1200 East California Blvd., CA 91125, United States. Email:{\tt taoyin89@caltech.edu}}}
\end{titlepage}
\maketitle
%\vspace{.2in}

\begin{abstract}
This paper is concerned with the boundary integral equation method for solving the exterior Neumann boundary value problem of dynamic poroelasticity in two dimensions. The main contribution of this work consists of two aspescts: the proposal of a novel regularized boundary integral equation, and  the presentation of new regularized formulations of  the strongly-singular and hyper-singular boundary integral operators. Firstly, turning to the spectral properties of the double-layer operator and the corresponding Calder\'{o}n relation of the poroelasticity, we propose the novel low-GMRES-iteration integral equation whose eigenvalues are bounded away from zero and infinity. Secondly, with the help of the G\"{u}nter derivatives, we reformulate the strongly-singular and hyper-singular integral operators into combinations of the weakly-singular operators and the tangential derivatives.  The accuracy and efficiency of the proposed methodology are demonstrated through  several numerical examples.

\end{abstract}

\section{Introduction}
\label{sec1}
We investigate the application of the boundary integral equation method (BIEM) to solve the dynamic poroelastic scattering problem~\cite{CD95,MB89,MS12,S011,S012,SSU09} in an unbounded exterior domain, and this problem is of great importance in many fields of applications such as oil and gas exploration, materials science, seismic analysis, etc. The poroelastic problem can be characterized by the Biot model~\cite{B41,B55,B561,B562,B563,B00,CBB91}, and the Neumann boundary condition will be considered in this work. Compared with the volumetric discretization methods ~\cite{DR93,DE96,LS98,XOX19}, the BIEM possesses such advantages as requiring discretization of domains of lower dimensionality and taking into account the radiation condition at infinity in a direct manner, and it has been widely studied for the numerical solutions of scattering problems~\cite{BXY17,BXY191,BXY192,BET12,BY20,CKM00,CN02,CBB91,YHX17}.

Regarding to the time-harmonic elastic wave scattering problems in an unbounded exterior domain, a combination form~\cite{BXY17,BY20,CK98} of single-layer and double-layer potentials is usually used to represent the solution, and the resulting combined boundary integral equation (CBIE) potentially ensures the validity of unique solvability corresponding to all frequencies. Although the unique solvability of the CBIE for the poroelastic scattering problem  remains unsolved, the CBIE  still provides an efficient numerical tool for the  solution of the problem imposed on the unbounded domain. Since the integral operators resulting from the action of the traction operator on both the single-layer and the double-layer potentials contain strongly-singular and hyper-singular kernels, they are only well defined in the sense of Cauchy principle value and Hadamard finite part~\cite{HW08}, respectively. Meanwhile, the appearance of the strongly-singular and hyper-singular integral operators in the CBIE leads to some difficulties related to the spectral character and the accurate evaluation of these operators. Firstly, it is known that the eigenvalues of the hyper-singular operators accumulate at infinity. Therefore, solving the CBIE by means of Krylov-subspace iterative solvers,  for instance the GMRES, often requires a relatively large number of iterations for the convergence of numerical solution. Secondly, the evaluation of the associated Cauchy principle value (resp. Hadamard finite part) of the strongly-singular (resp. hyper-singular) integrals has also been remaining a significant challenge.

To reduce the number of GMRES iterations required in the process of solving the CBIE, an efficient methodology which was originally proposed in~\cite{BET12,CN02} for acoustic and electromagnetic scattering problems utilizes the Calder\'on relation together with a regularized operator with a form similar to a single-layer operator. The derived regularized boundary integral equations (RBIEs) are of the second-kind Fredholm type. This approach has been extended to the homogeneous elastic cases~\cite{BXY192,BY20}, and it can be shown that the eigenvalues of the RBIE are bounded away from zero and infinity. However, the poroelastic Calder\'on formulas have not been studied in open literatures, and the main difficulty comes from the fact that the poroelastic double-layer integral operator (which plays important roles in the Calder\'on relations) is not compact. On the basis of the spectral property studied in~\cite{AJKKY18} for the static-elastic double-layer operator, it can be proved (see Theorem~\ref{spectra}) that the poroelastic double-layer integral operator is polynomially compact, and then the composition of the single-layer and hyper-singular integral operators can be expressed as the sum of a multiple of the identity operator and a compact operator. Relying on the theoretical results, we propose a new RBIE method  for solving the dynamic poroelastic problem,  and verify  numerically that the eigenvalues of the RBIE accumulate at fixed points only depending on the Lam\'e parameters of elastic media.

In this work, the classical Nystr\"om method, which has been widely used for the acoustic and elastic problems~\cite{CK98,CKM00,K95}, is employed for the numerical implementation of the proposed RBIE.   As applying the method, we encounter the challenge of accurate evaluation of the strongly-singular and hyper-singular integrals. In light of the novel regularized formulations presented in~\cite{BXY191,YHX17} for the elastic and thermoelastic problems, it can be shown (see Lemmas~\ref{regKP}-\ref{regN4}) that the strongly-singular and hyper-singular integrals can be re-expressed as compositions of weakly-singular integrals and tangential-derivative operators  by means of the G\"unter derivative and integration-by-parts. As a result, the Nystr\"om method allows us to evaluate the weakly-singular integrals with spectral accuracy for analytic surfaces, and to calculate the tangential-derivative of a given function via fast Fourier transform (FFT) in GMRES iterations. Numerical tests show that the proposed scheme demonstrate a simpler and more efficient performance than some alternative numerical treatments~\cite{K95,CKM00}.

The remainder of this paper is organized as follows. Section~\ref{sec:2} describes the dynamic poroelastic problem (Section~\ref{sec:2.1}) and the classical combined field integral equation (Section~\ref{sec:2.2}). Section~\ref{sec:3.1} theoretically and numerically studies the spectral properties of the poroelastic integral operators and the corresponding Calder\'on relation, and then a new regularized integral equation is proposed in Section~\ref{sec:3.2}. Exact reformulations of the strongly-singular and hyper-singular operators are presented in Section~\ref{sec:3.3}. The Nystr\"om method for numerical evaluation of the integral operators is briefly described in Section~\ref{sec:4}. Section~\ref{sec:5} presents the numerical examples to demonstrate the high-accuracy and efficiency of the proposed method. Finally, we present a  conclusion in Section~\ref{sec:6}.

\section{Preliminaries}
\label{sec:2}

\subsection{Poroelastic problem}
\label{sec:2.1}
Let $\Omega  \subset \R^2$ be a bounded domain with smooth boundary $\Gamma : = \partial \Omega $. Assume that the exterior domain ${\Omega ^c} = \R^2\backslash \overline \Omega   \subset \R^2$ is occupied by a linear isotropic poroelastic medium. Following the Biot's theory~\cite{B41,B561,B562} to model wave propagation in poroelastic medium, the dynamic poroelastic problem in frequency-domain to be considered in this work is characterized by the governed equations of the solid displacements $u=(u_1,u_2)^\top$ and the pore pressure $p$ that are given by
\begin{equation}
\label{model}
\begin{split}
& \Delta^*u + (\rho-\beta\rho _f)\omega ^2u - (\alpha-\beta)\nabla p = 0\\
&\Delta p + qp + i\omega\gamma\nabla\cdot u = 0
\end{split}\quad \mbox{in}\quad\Omega^c,
\end{equation}
or in an operator notation
\ben
LU=0, \quad L=\begin{bmatrix}
\Delta^*  + (\rho-\beta\rho _f)\omega ^2I   & -(\alpha -\beta )\nabla \\
i\omega\gamma\nabla \cdot & \Delta  + q
\end{bmatrix}, \quad U=(u_1,u_2,p)^\top,
\enn
where
\ben
\beta=\frac{{\omega {\phi^2}{\rho_f}\kappa}}{{i{\phi ^2} + \omega \kappa ({\rho _a} + \phi {\rho _f})}},\quad q = \frac{\omega^2\phi^2\rho_f}{\beta R},\quad
\gamma =  - \frac{{i\omega {\rho _f}(\alpha  - \beta )}}{\beta }.
\enn
Here, $\omega$ denotes the frequency, $I$ is the identity operator and $\Delta^*$ is the Lam\'e operator defined by
\ben
\Delta^*: = \mu\Delta + (\lambda  + \mu )\nabla\nabla\cdot
\enn
with $\Delta$ being the Laplacian operator, and  $\nabla$ being the gradient operator. The material parameters used in (\ref{model}) are listed in Table~\ref{Tablemodel}. In addition, we consider the Neumann boundary condition on $\Gamma$ given by
\be
\label{boundary condtion}
\widetilde{T}(\partial ,\nu )U: = \begin{bmatrix}
	T(\partial ,\nu ) & - \alpha\nu \\
	- i\omega\beta\nu^\top & \frac{i\beta}{\omega\rho_f}\partial_\nu
	\end{bmatrix}U = F,
\en
in which the traction operator $T(\partial ,\nu )$ is defined as
\ben
T(\partial ,\nu )u: = 2\mu{\partial _\nu }u + \lambda \nu\nabla\cdot u + \mu\nu^\perp  (\pa_2u_1-\pa_1u_2),\quad \nu^\perp=(-\nu_2,\nu_1)^\top,
\enn
where $\nu=(\nu_1,\nu_2)^\top$ denotes the outward unit normal to the boundary $\Gamma$ and $\pa_\nu:=\nu\cdot\nabla$ is the normal derivative. If the scattered field is induced by an incident field $U^{inc}$, the boundary data is determined as $F=-\widetilde{T}(\partial ,\nu )U^{inc}$.

\begin{table}[htbp]
\caption{The material parameters in poroelasticity.}
\centering
\begin{tabular}{|c|c|}
\hline
Notation & Physical meaning \\
\hline
$\lambda,\mu (\mu>0,\lambda+\mu>0)$ & Lam\'e parameters \\
$\nu_p$ & Poisson ratio \\
$\nu_u$ & undrained Poisson ratio \\
$B$ & Skempton porepressure coefficient \\
$\rho_s$ & solid density \\
$\rho_f$ & fluid density \\
$\rho_a$ & apparent mass density \\
$\phi$ & porosity \\
$\kappa$ & permeability coefficient \\
$\rho=(1-\phi)\rho_s + \phi\rho_f$ & bulk density \\
$\alpha=\frac{3(\nu_u-\nu_p)}{B(1-2\nu_p)(1+\nu_u)}$ &  compressibility \\
$R=\frac{2\phi^2\mu B^2(1-2\nu_p)(1+\nu_u)^2}{9(\nu_u-\nu_p)(1-2\nu_u)}$ & constitutive coefficient \\
\hline
\end{tabular}
\label{Tablemodel}
\end{table}

\subsection{Boundary integral equation}
\label{sec:2.2}

Let $E(x,y)$ be the fundamental solution of the adjoint operator $L^*$ of $L$ in $\R^2$ given by
\ben
E(x,y)=
\begin{bmatrix}
E_{11}(x,y) & E_{12}(x,y) \\
E_{21}^\top(x,y) & E_{22}(x,y)
\end{bmatrix},\quad x\ne y,
\enn
with
\ben
&& {E_{11}}(x,y) = \frac{1}{\mu}{\gamma _{{k_s}}}(x,y)I + \frac{1}{{(\rho  - \beta {\rho _f})}{\omega ^2}}{\nabla _x}\nabla _x^\top\left[ {{\gamma _{{k_s}}}(x,y) - {\frac{k_{p}^{2}-k_{2}^{2}}{k_{1}^{2}-k_2^2}}{\gamma _{{k_1}}}(x,y) + {\frac{k_{p}^{2}-k_{1}^{2}}{k_{1}^{2}-k_{2}^{2}}}{\gamma _{{k_2}}}(x,y)} \right],\\
&& {E_{12}}(x,y) =  \frac{i\omega\gamma}{(\lambda+2\mu)(k_{1}^{2}-k_{2}^{2})}{\nabla _x}[ {{\gamma _{{k_1}}}(x,y) - {\gamma _{{k_2}}}(x,y)} ],\\
&& {E_{21}}(x,y) =  -\frac{\gamma}{(\lambda+2\mu)(k_{1}^{2}-k_{2}^{2})}{\nabla _x}[ {{\gamma _{{k_1}}}(x,y) - {\gamma _{{k_2}}}(x,y)} ],\\
&& {E_{22}}(x,y) = \frac{i\rho_f\omega}{{\beta(k_1^2 - k_2^2)}}[ {(k_p^2 - {k_1 ^{2}}){\gamma _{{k_1}}}(x,y) - (k_p^2 - {k_2 ^{2}}){\gamma _{{k_2}}}(x,y)} ],
\enn
in which
\ben
\gamma _{k_t}(x,y) = \frac{i}{4}H_0^{(1)}({k_t}\left| {x - y} \right|), \quad x \ne y, \quad t=s,p,1,2,
\enn
denotes the fundamental solution of the Helmholtz equation in $\R^2$ with wave number $k_t$. Here, $k_p$ and $k_s$, referred as the compressional and shear wave numbers, respectively, are given by
\ben
k_p:=\omega\sqrt{\frac{\rho-\beta\rho_f}{\lambda  + 2\mu}}, \quad k_s:=\omega\sqrt{\frac{\rho-\beta\rho_f}{\mu}}.
\enn
The wave numbers $k_1$, $k_2$, satisfying $\mbox{Im}(k_i)\ge 0, i=1,2$, are the roots of the characteristic system
\ben
k_{1}^2+k_{2}^{2}=q(1+ \epsilon )+k_p^{2}, \quad k_1^{2}k_2^{2}=qk_p^{2},\quad \epsilon=\frac{i\omega\gamma(\alpha-\beta)}{q(\lambda+2\mu)},
\enn
and it follows that
\ben
k_1 &=& \sqrt {\frac{1}{2}\left\{ {k_p^2 + q(1 + \varepsilon ) + \sqrt {[ {k_p^2 + q(1 + \varepsilon )}] - 4q k_p^2} } \right\}}, \\
k_2 &=& \sqrt {\frac{1}{2}\left\{ {k_p^2 + q(1 + \varepsilon ) - \sqrt {[ {k_p^2 + q(1 + \varepsilon )} ] - 4qk_p^2} } \right\}}.
\enn

From the potential theory, the unknown function $U$ in $\Omega^c$ can be represented
as a combination of the single-layer and double-layer potentials
\be
U(x)=(D-i\eta S)(\varphi)(x), \qquad x\in \Omega^c, \quad Re(\eta) \ne 0,
\label{solrep}
\en
where
\be
\label{singlelayer}
\mathcal{S}(\varphi)(x):=\int_\Gamma  (E(x,y))^\top\varphi (y)ds_y,
\en
\be
\label{doublelayer}
\mathcal{D}(\varphi)(x):=\int_\Gamma (\widetilde T^*(\pa_y,\nu_y)E(x,y))^\top \varphi(y)ds_y,
\en
denote the single-layer and double-layer potentials, respectively. Here $\widetilde T^*$ denotes the corresponding Neumann boundary operator of $L^*$ given by
\be
\widetilde T^ * (\partial ,\nu ) = \begin{bmatrix}
	{T(\partial ,\nu )}&-{ i\omega \alpha \nu }\\
	{ - \beta {\nu ^\top}}&{\frac{{i\beta }}{{\omega {\rho _f}}}{\partial _\nu }}
	\end{bmatrix}.
\label{Tadjoint}
\en
The combination form of solution representation (\ref{solrep}) has been widely used for the corresponding acoustic and elastic scattering problems~\cite{BY20,CK98}, and the resulting boundary integral equation can ensure unique solvability for all frequencies. Operating with the boundary operator $\widetilde T$ on (\ref{solrep}), taking the limit as $x\rightarrow\Gamma$, the CBIE
\be
[i\eta(\frac{I}{2}-K^{\prime})+N](\varphi)=F \quad \mbox{on} \quad \Gamma
\label{BIE1}
\en
is obtained. Here $I$ denotes the identity operator and the boundary integral operators $K'$ and $N$ are defined by
\be
K'(\varphi)(x) =\widetilde T({\partial_x},{\nu_x}) \int_\Gamma  {(E(x,y))^{\top}\varphi (y)d{s_y}},
\label{KP}
\en
and
\begin{equation}
N(\varphi)(x) = \widetilde T({\partial_x},{\nu_x})\int_\Gamma  {{{\left({{\widetilde T}^*}({\partial_y},{\nu_y})E(x,y)\right)}^\top}\varphi (y)d{s_y}},
\label{N}
\end{equation}
in the sense of Cauchy principal value and Hadamard finite part~\cite{HW08}, respectively, in view of the strongly singular and hyper-singular character of the corresponding kernels.

\begin{remark}
\label{BIEremark}
It is well-known that the unkown solution $U$ in $\Omega^c$ can also be expressed simply as a single-layer potential
\be
\label{solrep1}
U(x) = \mathcal{S}(\psi)(x),\quad x \in {\Omega^c},
\en
or as a double layer potential
\be
\label{solrep2}
U(x)=\mathcal{D}(\psi)(x),\quad x \in {\Omega^c}.
\en
Operating with the boundary operator $\widetilde T$ on (\ref{solrep1}) and (\ref{solrep2}), taking the limit as $x\rightarrow\Gamma$, we can obtain the following boundary integral equations
\be
\label{BIE2}
(-\frac{I}{2}+K')(\psi)=F \quad \mbox{on}\quad\Gamma,
\en
and
\be
\label{BIE3}
N(\psi)=F \quad \mbox{on}\quad\Gamma,
\en
for the unknown density $\psi$ in (\ref{solrep1}) and (\ref{solrep2}), respectively.
\end{remark}

\begin{remark}
Unfortunately, the unique solvability of integral equation (\ref{BIER}) can not be derived following the classical approach to prove the corresponding unique solvability of combined field integral equations for acoustic and elastic problems~\cite{CK98,BXY17}. The main reason is that due to the special Neumann boundary operator $\widetilde{T}(\pa,\nu)$, there is no appropriate Green's first identity for the considered poroelastic problems and only the following Green's second identity holds
\ben
\int_\Omega \left(LU\cdot V-U\cdot L^*V\right)\,dx= \int_\Gamma \left( \widetilde{T}(\pa,\nu)U\cdot V-U\cdot \widetilde{T}^*(\pa,\nu)V\right)\,ds.
\enn
The uniqueness of integral equation (\ref{BIER}) still remains open, however, as discussed in Section~\ref{sec:5}, the determinant of the stiffness matrix resulting from the discretization of (\ref{BIER}) does not contain any significant shape trough which generally can indicate the existence of eigenfrequency~\cite{YHX17}.
\end{remark}

\begin{figure}[htb]
\centering
\includegraphics[scale=0.2]{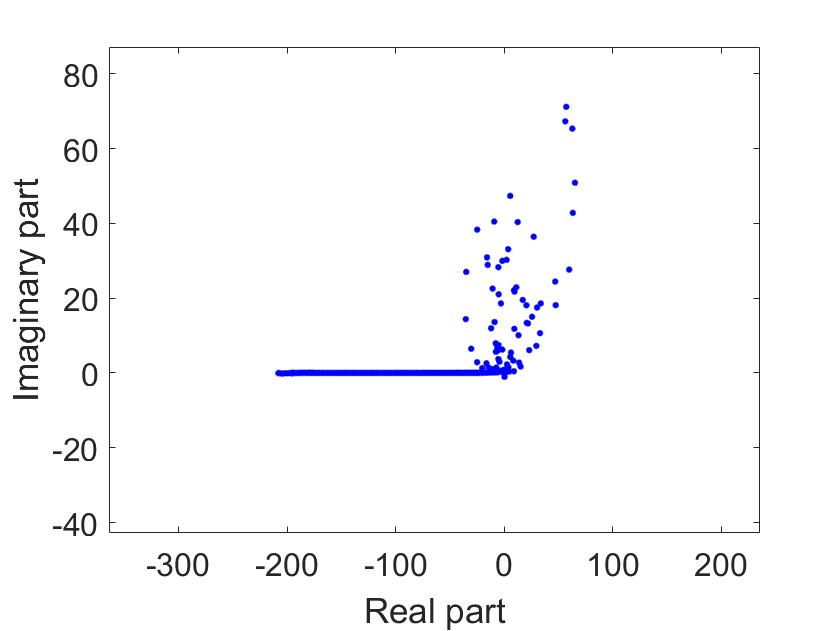}
\caption{Eigenvalue distribution of the operator $N$ for a circular scatterer.}
\label{Neig}
\end{figure}

\begin{remark}
Note that the eigenvalues of the hypersingular integral operator $N$ accumulate at infinity, see Figure~\ref{Neig} in which $\Gamma=\{|x|=1,x\in\R^2\}$ is considered. As a result, the solution of the integral equation (\ref{BIE1}) and (\ref{BIE3}) by means of Krylov-subspace iterative solvers such as GMRES generally requires large numbers of iterations.
\end{remark}

\section{Operator spectrum and regularized boundary integral equation}
\label{sec:3}

To avoid the difficulty arising from the use of hypersingular integral operator in (\ref{BIE1}), we propose a new RBIE   for solving the poroelastic problem in Section~\ref{sec:2.1}. Here, two types of ``regularization'' are employed. By means of introducing a regularized operator $R$ and studying the spectrum properties of poroelastic integral operators (Section~\ref{sec:3.1}), we derive a new boundary integral equation which corresponds to a linear system  with a better convergence property  after the discretization compared to that of (\ref{BIE1}), see Section~\ref{sec:3.2}. Meanwhile, the strongly-singular and hyper-singular integral operators are re-expressed into compositions of weakly-singular operators and differentiation operators in directions tangential to the boundary using the G\"unter derivative and integration by parts (Section~\ref{sec:3.3}) and we call this as a regularization procedure.

\subsection{operator spectrum}
\label{sec:3.1}

The spectra of the integral operators $K'$ is concluded in the following theorem.
\begin{theorem}
\label{spectra}
Let $\Gamma$ denote a smooth closed surface in two-dimensional space. Then $K'^2 - \begin{bmatrix}
 C_{\lambda ,\mu }^2I & 0\\
 0 & 0
 \end{bmatrix}: H^{1/2}(\Gamma)^3\rightarrow H^{1/2}(\Gamma)^3$ is compact, where $C_{\lambda,\mu}$ is a constant given by
\ben
C_{\lambda ,\mu }= \frac{\mu }{2(\lambda  + 2\mu )} < \frac{1}{2}.
\enn
Furthermore, the spectrum of $K^{\prime}$ consists of three nonempty sequences of eigenvalues which accumulate at $0$, $C_{\lambda,\mu}$ and $-C_{\lambda,\mu}$ respectively.
\end{theorem}
\begin{proof}
Recall the definition of $K'$ that
\ben
K'(U)(x)& =& \widetilde T(\partial _x,\nu _x)\int_\Gamma (E(x,y))^\top U(y)ds_y \\
& =& \begin{bmatrix}
	K'_1& K'_2\\
	K'_3 & K'_4
	\end{bmatrix}\begin{bmatrix}
	u\\
	p
	\end{bmatrix}(x),\quad x\in\Gamma,
\enn
where the operators $K'_j,j=1,\cdots,4$ are denoted as
\ben
K'_1(u)(x) &=& \int_\Gamma  \left(T(\partial_x,\nu_x)E_{11}-\alpha \nu_x E_{12}^\top\right) u(y)ds_y, \\
K'_2(p)(x) &=& \int_\Gamma  \left(T(\partial_x,\nu_x)E_{21}-\alpha \nu_x E_{22}\right) p(y)ds_y, \\
K'_3(u)(x)& =& \int_\Gamma  \left(-i\omega\beta\nu_x^{\top}E_{11}+ \frac{i\beta}{\rho_f\omega}\partial_{\nu_x}E_{12}^\top\right)u(y)ds_y, \\
K'_4(p)(x)& =& \int_\Gamma  \left(-i\omega\beta\nu_x^{\top}E_{21} +\frac{i\beta}{\rho_f\omega}\partial_{\nu_x}E_{22}\right)p(y)ds_y.
\enn
Let $K_0^\prime$ denote the static $(\omega=0)$ boundary integral operator corresponding to $K^\prime$
\ben
K_0' &=& \widetilde T_0(\partial _x,\nu _x)\int_\Gamma E_0(x,y) U(y)ds_y\\
&=& \left[ \begin{array}{*{20}{c}}
	K'_{1,0}&K'_{2,0}\\
	K'_{3,0}&K'_{4,0}
	\end{array} \right]\left[ {\begin{array}{*{20}{c}}
	u\\
	p
	\end{array}} \right](x), \quad x\in\Gamma
\enn
where
\ben
\widetilde T_0(\partial _x,\nu _x)=\left[ \begin{array}{*{20}{c}}
	T(\partial _x,\nu _x)&-\alpha\nu_x\\
	0&\kappa\partial_{\nu_x}
	\end{array} \right]
\enn
and
\ben
E_{0}(x,y)=\left[ \begin{array}{*{20}{c}}
	E_{0,11}&E_{0,12}\\
	E_{0,21}&E_{0,22}
\end{array} \right]=\left[ \begin{array}{*{20}{c}}
	E_{e,0}&-\frac{\alpha(x-y)}{2(\lambda+2\mu)}\ln\left|x-y\right|\\
	0&-\frac{1}{2\pi}\ln\left|x-y\right|
	\end{array} \right]
\enn
is the fundamental solution of static poroelastic problem with $E_{e,0}(x,y)$ being the fundamental solution of Lam\'e equation which is given by
\begin{equation*}
E_{e,0}(x,y)=\frac{\lambda+3\mu}{4\pi\mu(\lambda+2\mu)}\left\{ -\ln|x-y|I+\frac{\lambda+\mu}{\lambda+3\mu}\frac{1}{|x-y|^2}(x-y)(x-y)^\top \right\}.
\end{equation*}
Thus we can obtain that
\ben
K'_{1,0}(u)(x) &=& \int_\Gamma  T(\partial_x,\nu_x)E_{0,11} u(y)ds_y, \\
K'_{2,0}(p)(x) &=& \int_\Gamma  \left(T(\partial_x,\nu_x)E_{0,12}-\alpha\nu_xE_{0,22}\right) p(y)ds_y, \\
K'_{3,0}(u)(x) &=& 0, \\
K'_{4,0}(p)(x) &=& \int_\Gamma  \kappa \pa_{\nu_x}E_{0,22} p(y)ds_y. \\
\enn
From \cite{AJKKY18}, it is known that $K_{1,0}^{'2}-C_{\lambda,\mu}^{2}I$ is compact. It can be easily deduced that the kernels of $K'_{j,0},j=2,4$ are weakly-singualar implying that $K'_{j,0},j=2,4$ are compact. Therefore,
\ben
 K^{\prime 2} -\begin{bmatrix}
 C_{\lambda ,\mu }^2I & 0\\
 0 & 0
 \end{bmatrix}
=K'(K' - K'_0)+ (K' - K'_0)K'_0 +\begin{bmatrix}
 {K_{1,0}^{\prime2} - C_{\lambda ,\mu }^2I  } & K'_{1,0}K'_{2,0}+K'_{2,0}K'_{4,0}\\
 0 & K_{4,0}^{\prime2}
 \end{bmatrix}
\enn
is compact due to the fact that $K' - K'_0$ has a weakly-singular kernel and is a compact operator. The inequality $0<C_{\lambda,\mu}<1/2$ can be obtained easily from the conditions $\lambda+\mu>0,\mu>0$. This completes the proof.
\end{proof}

From Theorem~\ref{spectra}, it can be seen that the accumulation points of the eigenvalues of $K'$ are independent of the frequency. In addition, we can conclude from the Calder\'{o}n relation
\begin{equation}
NS=-\frac{I}{4}+K'^{2}
\end{equation}
that the spectrum of the composite operator $NS$, which plays an essential role in the regularized integral equations proposed in the following section, consists of two nonempty sequences of eigenvalues which accumulate at $-\frac{1}{4}$ and $-\frac{1}{4}+C_{\lambda,\mu}^2$.

In order to verify the above results numerically, we consider the problem of poroelastic scattering by a circular scatterer of radius one, and choose the same values of coefficients as in Section~\ref{sec:5} which gives $C_{\lambda,\mu}=0.132$. Figure~\ref{KNSeig} displays the eigenvalue distribution of the integral operators $K'$ and $NS$ from which the eigenvalues of $K'$ and $NS$ are seen to accumulate at the points predicted by our theoretical results. Here, the eigenvalue computation, on a basis of the exact regularized formulations for the strongly-singular operator $K'$ and the hypersingular operator $N$ given in Section~\ref{sec:3.3}, has been implemented by means of the high-order Nystr\"om methodology (see Section~\ref{sec:4}) together with FFT for evaluation of tangential derivatives and choosing a sufficiently large number of discretization points.

\begin{figure}[htb]
\centering
\begin{tabular}{cc}
\includegraphics[scale=0.2]{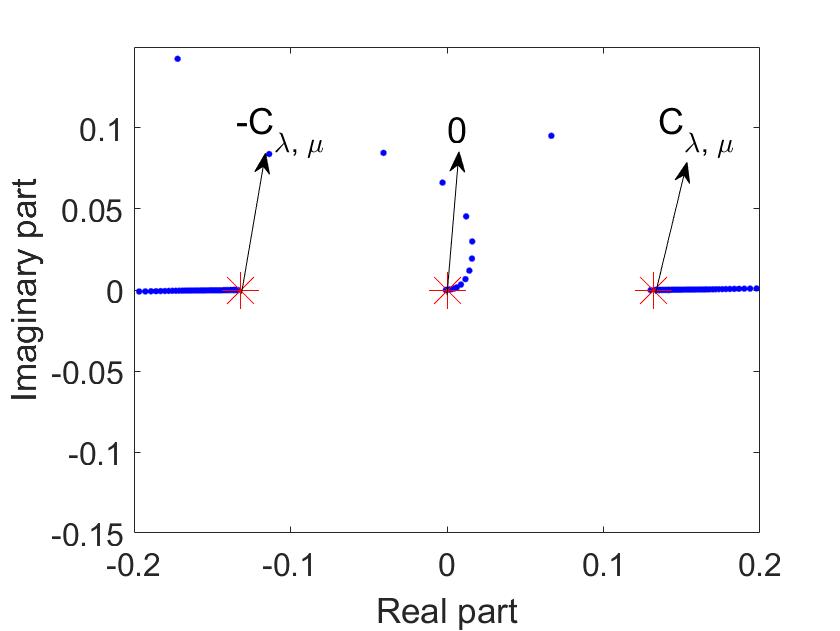} &
\includegraphics[scale=0.2]{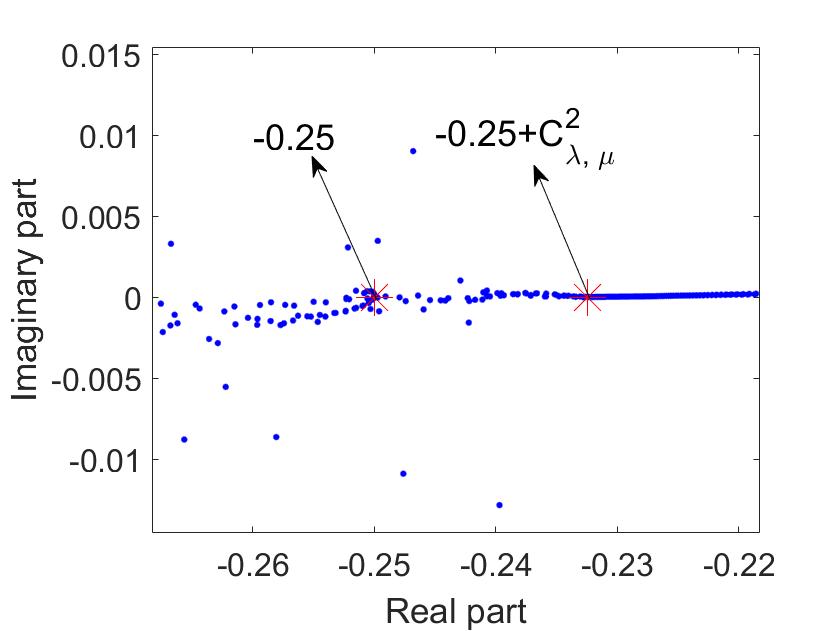} \\
(a) $K'$ & (b) $NS$
\end{tabular}
\caption{Eigenvalue distribution of the operators $K'$ and $NS$ for a circular scatterer.}
\label{KNSeig}
\end{figure}

\subsection{Regularized boundary integral equation}
\label{sec:3.2}

Relying on the spectra studies of the poroelastic integral operators presented in Section~\ref{sec:3.1}, we propose in this section a RBIE by utilizing a regularization operator $\mathcal{R}$ in addition to the aforementioned double-layer and hypersingular operators $K'$ and $N$. In light of the regularized integral equation method discussed in~\cite{BET12} for acoustic problems, $\mathcal{R}$ is specified as the static ($\omega=0$) single-layer operator $S_0$ corresponding to $S$ in this section (see the comparison between $\mathcal{R}=S_0$ and other selections of the regularization operator $R$ in Section~\ref{sec:5}). Then replacing the solution presentation (\ref{solrep}) by
\be
U(x)=(D\mathcal{R}-i\eta S)(\varphi)(x), \quad x\in \Omega^c.
\label{solrepR}
\en
we obtain the following RBIE
\be
\left[ i\eta(\frac{I}{2}-K^\prime)+N\mathcal{R} \right](\varphi)=F \quad \mbox{on} \quad \Gamma,
\label{BIER}
\en
instead of the classical CBIE (\ref{BIE1}). It follows from the spectra results in Section~\ref{sec:3.1} that the spectrum of the regularized combined field integral operator on the left hand side of (\ref{BIER}) consists of three non-empty sequences of eigenvalues which converge to $-1/4+i\eta/2$, $-1/4+C^2_{\lambda,\mu}+i\eta(1/2+C_{\lambda,\mu})$ and $-1/4+C^2_{\lambda,\mu}+i\eta(1/2-C_{\lambda,\mu})$, respectively, and these values are all bounded away from zero and infinity, see Figure~\ref{KPNSeig}.

\begin{figure}[htb]
\centering
\includegraphics[scale=0.2]{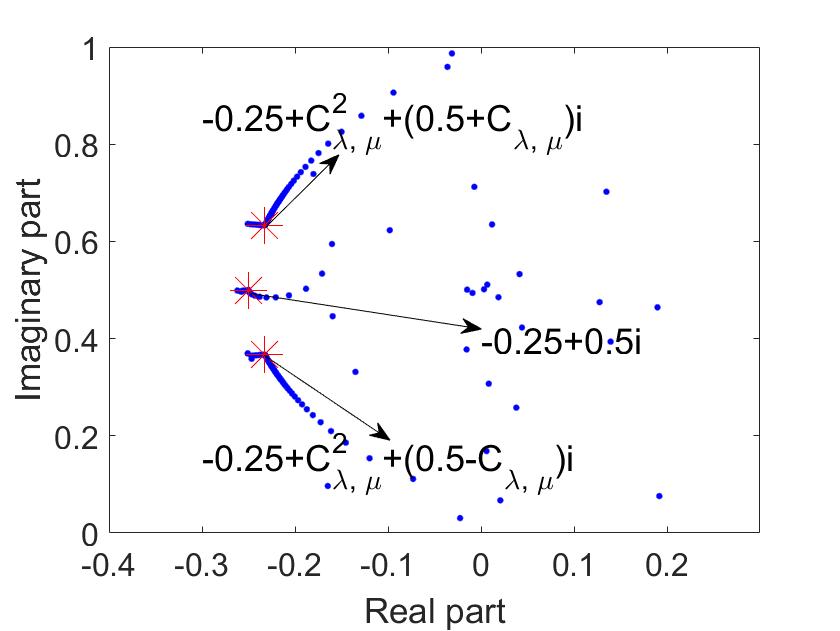}
\caption{Eigenvalue distribution of the integral operator $i(\frac{I}{2}-K^{\prime})+NS$ for a circular scatterer.}
\label{KPNSeig}
\end{figure}

\subsection{Strong-singularity and hyper-singularity regularization}
\label{sec:3.3}

As aforementioned, the integral operators $K'$ and $N$ are strongly-singular and hyper-singular, respectively. In this section, following the techniques proposed in~\cite{YHX17}, we re-expresses these two operators as combinations of weakly-singular integral operators and tangential derivative. (The main approach to derive the regularized formulations will be shown in Appendix.) Applying these new formulations together with the Nystr\"{o}m method to be described in Section~\ref{sec:4} and the linear algebra solver GMRES to equation (\ref{BIER}) would lead  then leads to the proposed solvers for the boundary value problem in Section~\ref{sec:2.1}.

From~\cite{YHX17}, it is known that the traction operator $T(\partial ,\nu )$ can be rewritten as
\ben
T(\partial ,\nu )u(x) = (\lambda  + 2\mu )\nu (\nabla  \cdot u) + \mu {\partial _\nu }u + \mu M(\partial ,\nu )u,
\enn
where the operator $M(\partial ,\nu )$, whose elements are also called G\"{u}nter derivatives\cite{HW08}, admits
\ben
M(\partial ,\nu )u(x)= A\frac{du}{ds},\quad A = \begin{bmatrix}
	0&{ - 1}\\
	1&0
	\end{bmatrix}.
\enn
For the strongly-singular integral operator $K'$, the following regularized formulation can be obtained using the notations in Theorem~\ref{spectra}.

\begin{lemma}
\label{regKP}
The boundary integral operators $K'_j,j=1,\cdots,4$ can be expressed as
\be
\label{KPP}
K'_j=K_j^{\prime1}+\frac{d}{ds}K_j^{\prime2},\quad j=1,\cdots,4,
\en	
where
\ben
K_1^{\prime1}(u)(x) &=& \int_\Gamma  \nu _x\nabla _x^\top\left[ - (\gamma _{k_s}(x,y) - \gamma _{k_1}(x,y)) + \frac{k_2^2 - q}{k_1^2 - k_2^2}(\gamma _{k_s}(x,y) - \gamma _{k_1}(x,y))\right]u(y) ds_y \nonumber\\
&\quad& + \int_\Gamma  \left[\partial_{\nu_x}\gamma_{k_s}(x,y)I - \alpha \nu_xE_{12}^\top(x,y) \right] u(y)ds_y,\\
K_1^{\prime2}(u)(x) &=& A\int_\Gamma  \left[2\mu E_{11}(x,y) - \gamma _{k_s}(x,y)I\right] u(y)ds_y,\\
K_2^{\prime1}(p)(x) &=& \int_\Gamma  \left[\frac{\gamma }{{k_1^2 - k_2^2}} (k_1^2\gamma _{k_1}(x,y) - k_2^2\gamma _{k_2}(x,y))\nu _x - \alpha \nu _xE_{22}(x,y)\right]p(y)ds_y,\\
K_2^{\prime2}(p)(x) &=& -\frac{2\mu\gamma}{(\lambda+2\mu)(k_1^2-k_2^2)}A\int_\Gamma \nabla _x\left[\gamma_{k_1}(x,y) - \gamma_{k_2}(x,y)\right]p(y)ds_y,\\
K_3^{\prime1}(u)(x) &=& \int_\Gamma  \left[ - i\omega \beta \nu _x^\top E_{11}(x,y) + C_0(k_1^2\gamma _{k_1}(x,y) - k_2^2\gamma_{k_2}(x,y))\nu _x^\top\right] u(y)ds_y,\\
K_3^{\prime2}(u)(x) &=& C_0\int_\Gamma  \nabla _x^\top\left[\gamma_{k_1}(x,y) - \gamma_{k_2}(x,y)\right]Ap(y)ds_y,\\
K_4^{\prime1}(p)(x) &=& \int_\Gamma  \nu_x^\top\left[i\omega \beta E_{21}(x,y) + \nabla _x\left(\frac{k_p^2 - k_1^2}{k_1^2 - k_2^2}\gamma_{k_1}(x,y) - \frac{k_p^2 - k_2^2}{k_1^2 - k_2^2}\gamma _{k_2}(x,y)\right)\right] p(y)ds_y,\\
K_4^{\prime2}(p)(x) &=& 0,
\enn
with the constant $C_0$ being
\ben
C_0=\frac{\beta \gamma}{\rho_{f}(\lambda+2\mu)(k^2_1-k^2_2)}.
\enn
\end{lemma}

Now we consider the hyper-singular operator $N$. For $\psi=(u^\top,p)^\top$, denote
\ben
N(\psi)(x)=\begin{bmatrix}
N_1 & N_2\\
N_3 & N_4
\end{bmatrix}\begin{bmatrix}
u\\
p
\end{bmatrix}(x),\quad x\in\Gamma.
\enn
The regularized formulations for the operators $N_j,j=1,\cdots,4$ are given in the following Lemmas~\ref{regN1}-\ref{regN4}.
\begin{lemma}
\label{regN1}
The hyper-singular operator $N_1$ can be expressed as
\be
\label{N1}
N_{1}=N^1_1+\frac{d}{ds}N^2_1\frac{d}{ds}+\frac{d}{ds}N^3_1+N^4_1\frac{d}{ds},
\en
where
\ben
N^1_{1}(u)(x) &= &-(\rho-\beta\rho_{f}) \omega^{2}\int _{\Gamma}\gamma_{k_{s}}(x,y)\left(\nu_{x}\nu_{y}^{\top} -\nu^{\top}_{x}\nu_{y}I-J_{\nu_{x},\nu_{y}}\right)u(y)ds_{y},\\ &\quad& +\int_{\Gamma}\left[C_{1}\gamma_{k_{1}}(x,y)-C_{2}\gamma_{k_{2}}(x,y)\right] \nu_{x}\nu_{y}^{\top}u(y)ds_{y}\\
N^2_1(u)(x)&= &4\mu\int_{\Gamma}\left[\gamma_{k_{s}}(x,y)+ AE_{11}(x,y)A\right]u(y)ds_y,\\
N^3_1(u)(x)&= & \int_\Gamma A\nabla _y \left[-2\mu(\gamma_{k_s}(x,y)-\gamma_{k_1}(x,y) )+C_3(\gamma_{k_1}(x,y)-\gamma_{k_2}(x,y) )\right] \nu^\top_y u(y)ds_y,\\
N^4_1(u)(x)&= &\int_\Gamma \nu_x\nabla _x^\top \left[-2\mu(\gamma_{k_s}(x,y)-\gamma_{k_1}(x,y) )+C_3(\gamma_{k_1}(x,y)-\gamma_{k_2}(x,y) )\right]Au(y)ds_y,
\enn
with the constant $C_{i}$, $i=1,2,3$ being
\ben
{C_1} = \frac{{k_1^2(k_1^2 - q)(\lambda  + 2\mu ) -2i\omega \alpha \gamma k^{2}_{1}}}{{ (k_1^2 - k_2^2)}}, \qquad {C_2} = \frac{{ k_2^2(k_2^2 - q)(\lambda  + 2\mu ) -2i\omega \alpha \gamma k^{2}_{2}}}{{ (k_1^2 - k_2^2)}}
\enn
\ben
{C_3} = \frac{{2\mu }}{{k_1^2 - k_2^2}}(-\frac{{i\omega \gamma \alpha }}{{\lambda  + 2\mu }} + k_2^2 - q).
	\enn
\end{lemma}

\begin{lemma}
\label{regN2}
The hyper-singular operator $N_2$ can be expressed as
\be
\label{N2}
N_2=N^1_2+\frac{d}{ds}N^2_2\frac{d}{ds}+\frac{d}{ds}N^3_2,
\en
where
\ben
	N^1_2(p)(x) &=& \beta \int_\Gamma  {{\nu _x}\nabla _x^{\top}\left[({\gamma _{{k_s}}}(x,y) - {\gamma _{{k_1}}}(x,y)) - \frac{{k_2^2 - q}}{{k_1^2 - k_2^2}}({\gamma _{{k_1}}}(x,y) - {\gamma _{{k_2}}}(x,y))\right]{\nu _y}p(y)} d{s_y}\\
	&\quad&- \beta \int_\Gamma  \left[{\partial _{{\nu _x}}}{\gamma _{{k_s}}}(x,y)  - \alpha {\nu _x}E_{12}^\top(x,y)\nu _y \right]p(y)d{s_y}\\
	&\quad& + \frac{{i\beta \gamma }}{{{\rho _f}\omega (k_1^2 - k_2^2)}}\int_\Gamma  \left[{{\partial _{{\nu _y}}}(k_1^2{\gamma _{{k_1}}}(x,y) - k_2^2{\gamma _{{k_2}}}(x,y))}\right] {\nu _x}p(y)d{s_y}\\
	&\quad&+\frac{\alpha }{{k_1^2 - k_2^2}}\int_\Gamma  \partial _{\nu _y}\left[(k_p^2 - k_1^2)k_1^2{\gamma _{{k_1}}}(x,y) - (k_p^2 - k_2^2)k_2^2{\gamma _{{k_2}}}(x,y) \right]{\nu _x}p(y)d{s_y},\\
	N^2_2(p)(x)&=&- \frac{{2i\mu \beta \gamma }}{{{\rho _f}\omega (\lambda  + 2\mu )(k_1^2 - k_2^2)}}\int_\Gamma  {{\nabla _x}\left[{\gamma _{{k_1}}}(x,y) - {\gamma _{{k_2}}}(x,y)\right]} d{s_y}, \\
	N^3_2(p)(x)&=&- \beta \int_\Gamma A\left[2\mu E_{11}(x,y) - \gamma _{k_s}(x,y)I)\right]\nu _yp(y)ds_y\\
	&\quad&- \frac{{2i\mu \beta \gamma }}{{{\rho _f}\omega (\lambda  + 2\mu )(k_1^2 - k_2^2)}}\int_\Gamma  \left[k_1^2{\gamma _{{k_1}}}(x,y) - k_2^2{\gamma _{{k_2}}}(x,y)\right]A {\nu _y}p(y)d{s_y}.
\enn
\end{lemma}

\begin{lemma}
\label{regN3}
The hyper-singular operator $N_3$ can be expressed as
\be
N_{3}=N^1_3+\frac{d}{ds}N^2_3\frac{d}{ds}+N^3_3\frac{d}{ds},
\label{N3}
\en
where
\ben
N^1_3(u)(x)&=&- i\omega \beta \int_\Gamma  {\nu _x^{\top
		}{\nabla _x}\left[{\gamma _{{k_s}}}(x,y) - {\gamma _{{k_1}}}(x,y)-\frac{{ k_2^2 - q}}{{k_1^2 - k_2^2}}({\gamma _{{k_1}}}(x,y) - {\gamma _{{k_2}}}(x,y))\right]\nu _y^{\top}u(y)} d{s_y}\\
	&\quad&- i\omega \beta \int_\Gamma  \left[{\partial _{{\nu _y}}}{\gamma _{{k_s}}}(x,y) \nu _x^{\top} - i\omega \alpha \nu _x^{\top}{E_{21}}\nu _y^{\top}\right]u(y)d{s_y}\\
	&\quad&- \frac{{\beta \gamma }}{{{\rho _f}(k_1^2 - k_2^2)}}\int_\Gamma  {{\partial _{{\nu _x}}}\left[k_1^2{\gamma _{{k_1}}}(x,y) - k_2^2{\gamma _{{k_2}}}(x,y)\right]} \nu _y^{\top}u(y)d{s_y}\\
	&\quad&+ \frac{{i\omega \alpha }}{{k_1^2 - k_2^2}}\int_\Gamma  {\partial _{{\nu _x}}}\left[(k_p^2 - k_1^2){\gamma _{{k_1}}}(x,y) - (k_p^2 - k_2^2){\gamma _{{k_2}}}(x,y) \right]\nu _y^{\top}u(y)d{s_y},\\
	N^2_3(u)(x)&=&- \frac{{2\mu \beta \gamma }}{{{\rho _f}(\lambda  + 2\mu )(k_1^2 - k_2^2)}}\int_\Gamma  {\nabla _x^\top\left[{\gamma _{{k_1}}}(x,y) - {\gamma _{{k_2}}}(x,y)\right]u(y)} d{s_y}\\
	N^3_3(u)(x)&=&- i\omega \beta \int_\Gamma  {\nu _x^{\top}\left[2\mu {E_{11}}(x,y) - {\gamma _{{k_s}}}(x,y)I\right]A} u(y)d{s_y}\\
	&\quad&+ \frac{{2\mu \beta \gamma }}{{{\rho _f}(\lambda  + 2\mu )(k_1^2 - k_2^2)}}\int_\Gamma  {\left[k_1^2{\gamma _{{k_1}}}(x,y) - k_2^2{\gamma _{{k_2}}}(x,y)\right]\nu _x^{\top}Au(y)} d{s_y}.\\
\enn
\end{lemma}

\begin{lemma}
\label{regN4}
The hyper-singular operator $N_4$ can be expressed as
\be
N_4=N^1_4+\frac{d}{ds}N^2_4\frac{d}{ds}+\frac{d}{ds}N^3_4,
\label{N4}
\en
where
\ben
	N^1_4(p)(x) &=& i\omega {\beta ^2}\int_\Gamma  {\nu _x^{\top}{E_{11}}(x,y)p(y)} d{s_y}+ \frac{{2i\beta\gamma\omega }}{{(\lambda  + 2\mu )}}C_4\int_\Gamma  {\left[k_1^2{\gamma _{{k_1}}}(x,y) - k_2^2{\gamma _{{k_2}}}(x,y)\right]\nu _x^{\top}{\nu _y}p(y)} d{s_y}\\
	N^2_4(p)(x) &=& \frac{{i\beta\gamma\omega }}{{(\lambda  + 2\mu )}}C_4\int_\Gamma  {\left[{\gamma _{{k_1}}}(x,y) - {\gamma _{{k_2}}}(x,y)\right]p(y)} d{s_y}\\
	&\quad&- C_4\int_\Gamma  {\left[(k_p^2 - k_1^2){\gamma _{{k_1}}}(x,y) - (k_p^2 - k_2^2){\gamma _{{k_2}}}(x,y)\right]p(y)} d{s_y}\\
	N^3_4(p)(x) &=& \frac{{i\beta\gamma\omega }}{{(\lambda  + 2\mu )}}C_4\int_\Gamma  {\nabla _x^\top[{\gamma _{{k_1}}}(x,y) - {\gamma _{{k_2}}}(x,y)]A} {\nu _y}p(y)d{s_y}\\
	&\quad& - C_{4}\int_\Gamma  {\left[(k_p^2 - k_1^2)k_1^2{\gamma _{{k_1}}}(x,y) - (k_p^2 - k_2^2)k_2^2{\gamma _{{k_2}}}(x,y)\right]\nu _x^{\top}{\nu _y}p(y)} d{s_y},\\
\enn
with
\ben
C_4=\frac{{i\beta }}{{{\rho _f}\omega (k_1^2 - k_2^2)}}.
\enn
\end{lemma}

\begin{remark}
It can be verified that the operators $K_j^{\prime i},i=1,2,j=1,\cdots,4$, $N_j^i, i=1,\cdots,3,j=1,\cdots,4$ and $N_1^4$ are all weakly-singular for smooth boundary $\Gamma$. The derived results of the regularized formulation for the integral operators $K'$ and $N$ can also be extended to the case of Lipshitz boundary in terms of the properties of G\"unter derivative given in~\cite{BT16}.
\end{remark}

\section{Numerical implementation: Nystr\"om method}
\label{sec:4}

According to the regularized formulations of integral operators $K'$ and $N$ given in Section~\ref{sec:3.3}, the numerical implementation of the RBIE (\ref{BIER}) can be converted into the evaluation of multiple operators of two types, (i) Integral operators of the form
\be
\mathcal{H}(\varphi)(x) = \int_\Gamma  {H(x,y)\varphi (y)} d{s_y}
\label{kernelH}
\en
in which the kernel $H(x,y)$ is weakly-singular, and (ii) Tangential derivative $d/ds$ of a given smooth function defined on $\Gamma$. This section presents algorithms for numerical evaluation of operators of these types by utilizing the well-known Nystr\"om method~\cite{CK98}.

Assume that the boundary curve $\Gamma$ is analytic and is given through
\be
\Gamma = \{x(t)=(x_1(t),x_2(t)):0\leq t\leq 2\pi\}
\label{boundary curve}
\en
in counterclockwise orientation where $x:\mathbb{R}\rightarrow \mathbb{R}^{2}$ is analytic and 2$\pi$-periodic with $\left | x'(t) \right |>0$ for all $t$. Then the outward unit normal at $x\in\Gamma$ is given by $\nu_t=\nu_{x(t)}=\frac{(x_2^{'}(t),-x_1^{'}(t))}{\left | x'(t) \right |}$ and the integral (\ref{kernelH}) can be transformed into a parametric form
\be
\mathcal{H}(\varphi)(x(t))=\int_0^{2\pi } {\widetilde H(t,\tau )\psi (\tau )} d\tau,
\label{parametric form}
\en
where $\widetilde H(t,\tau )=H(x(t),x(\tau))|x'(\tau)|$ and $\psi (t): = \varphi (x(t))$. For the weakly-singular kernel $\widetilde H(t,\tau )$, the Nystr\"om method requires splitting it into
\be
\widetilde H(t,\tau)=\widetilde H_1(t,\tau)\ln (4\sin^{2}\frac{t-\tau}{2})+\widetilde H_2(t,\tau),
\label{split kernel}
\en
where the terms $\widetilde H_1$, $\widetilde H_2$ are analytic. Choosing an equidistant mesh $t_j:=\frac{j\pi}{n}, j=0,\cdots,2N-1  (N\in \mathbb{N})$, the quantities  $H(\varphi)(x(t_i)), i=0,\cdots,2N-1$ can be approximated with exponential decaying errors using the weighted trigonometric interpolation quadratures and the trapezoidal rule, that is,
\ben
\mathcal{H}(\varphi)(x(t_i))\approx \sum_{j=0}^{2N-1}R_j^{(n)}(t_i)\widetilde H_1(t_i,t_j)\psi(t_j)+ \frac{\pi}{N}\sum_{j=0}^{2N-1}H_2(t_i,t_j)\psi(t_j).
\enn
For the regularized integral equation (\ref{BIER}), it can be summarized from the regularized formulations derived in Section~\ref{sec:3.3} that the weakly-singular kernel $\widetilde{H}(t,\tau)$ only takes four types of forms that are listed as follows, in which $c_e=0.57721566\cdots$ denotes the Euler constant, together with the corresponding splitting terms $\widetilde{H}_1(t,\tau)$, $\widetilde{H}_2(t,\tau)$.
\begin{itemize}
\item {\bf Type 1.}
\ben
\widetilde H(t,\tau) &=& H_0^{(1)}(k|x(t) - x(\tau )|)\left|x'(\tau)\right|, \\
\widetilde H_1(t,\tau) &=& \begin{cases}
\frac{i}{\pi }{J_0}(k\left| {x(t) - x(\tau )} \right|)\left|x'(\tau)\right|, & t\ne\tau, \cr
\frac{i}{\pi }\left|x'(t)\right|, & t=\tau,
\end{cases}\\
\widetilde H_2(t,\tau) &=& \begin{cases}
\widetilde H(t,\tau) - \widetilde H_1(t,\tau)\ln (4{\sin ^2}(\frac{{t - \tau }}{2})), & t\ne\tau, \cr
(1+\frac{2i}{\pi }(c_e+\ln\frac{k\left | {x(\tau )}^\prime \right |}{2}))\left|x'(t)\right|, & t=\tau.
\end{cases}
\enn
\item {\bf Type 2.}
\ben
&& \widetilde H(t,\tau) = \frac{\hat{k}H_1^{(1)}(\hat{k}|x(t) - x(\tau )|)-\tilde{k}H_1^{(1)}(\tilde{k}|x(t) - x(\tau )|)}{|x(t) - x(\tau )|}\left|x'(\tau)\right|,\\
&&\widetilde H_1(t,\tau) = \begin{cases}
\frac{i}{\pi}\frac{\hat{k}J_1(\hat{k}|x(t) - x(\tau )|)-\tilde{k}J_1(\tilde{k}|x(t) - x(\tau )|)}{\left| x(t)-x(\tau)\right|}\left|x'(\tau)\right|, & t\ne\tau, \cr
\frac{1}{2\pi }(k_1^2-k_2^2)\left|x'(t)\right|, & t=\tau,
\end{cases}\\
&&\widetilde H_2(t,\tau) = \begin{cases}
\widetilde H(t,\tau) - \widetilde H_1(t,\tau)\ln (4{\sin ^2}(\frac{{t - \tau }}{2})), & t\ne\tau, \cr
\left[\frac{k_{1}^{2}-k_{2}^{2}}{2}(1+\frac{2i}{\pi}(\ln\left | x^\prime(t) \right |+c_e)-\frac{i}{\pi})+\frac{2i}{\pi}\left(k_{1}^{2} \ln\frac{k_{1}}{2}-k_{2}^{2}\ln\frac{k_{2}}{2}\right)\right]\left|x'(t)\right|, & t=\tau.
\end{cases}
\enn
\item {\bf Type 3.}
\ben
&&\widetilde H(t,\tau) = \frac{\hat{k}^2H_2^{(1)}(\hat{k}|x(t) - x(\tau )|)-\tilde{k}^2H_2^{(1)}(\tilde{k}|x(t) - x(\tau )|)}{|x(t) - x(\tau )|^2}\left|x'(\tau)\right|,\\
&&\widetilde H_1(t,\tau) = \begin{cases}
\frac{i}{\pi}\frac{\hat{k}^2J_2(\hat{k}|x(t) - x(\tau )|)-\tilde{k}^2J_2(\tilde{k}|x(t) - x(\tau )|)}{\left| x(t)-x(\tau)\right|^2}\left|x'(\tau)\right|, & t\ne\tau \cr
0, & t=\tau
\end{cases}\\
&&\widetilde H_2(t,\tau) = \begin{cases}
\widetilde H(t,\tau) - \widetilde H_1(t,\tau)\ln (4{\sin ^2}(\frac{{t - \tau }}{2})), & t\ne\tau, \cr
-\frac{i}{\pi}(k_1^2-k_2^2)\left|x'(t)\right|, & t=\tau.
\end{cases}
\enn
\item {\bf Type 4.}
\ben
\widetilde H(t,\tau) &=& \frac{ik}{2}\nu_t\cdot \left[x(\tau)-x(t)\right] \frac{H_1^{(1)}(k\left|x(t)-x(\tau)\right|)}{\left|x(t)-x(\tau)\right|} \frac{\left|x'(\tau)\right|}{\left|x'(t)\right|},\\
\widetilde H_1(t,\tau) &=& \begin{cases}
	-\frac{k}{2\pi}\nu_t\cdot \left[x(\tau)-x(t)\right]\frac{J_1(k\left|x(t)-x(\tau)\right|)}{\left|x(t)-x(\tau)\right|}\frac{\left|x'(\tau)\right|}{\left|x'(t)\right|}, & t\ne\tau, \cr
	0, & t=\tau,
\end{cases}\\
\widetilde H_2(t,\tau) &=& \begin{cases}
	\widetilde H(t,\tau) - \widetilde H_1(t,\tau)\ln (4{\sin ^2}(\frac{{t - \tau }}{2})), & t\ne\tau, \cr
	\frac{1}{2\pi}\frac{\nu_t\cdot x^{\prime\prime}(t)}{\left|x'(t)\right|^2}, & t=\tau.
\end{cases}
\enn
\end{itemize}

Finally, for a given periodic smooth function $\varphi(x(t)), t\in[0,2\pi]$ with values $\varphi(x(t_j)), j=0,\cdots,2N-1$, applying the Fourier series approximation
\ben
\varphi(x(t))=\sum\limits_{m=-N}^{N-1} \varphi_m e^{imt}, \quad  t\in[0,2\pi],
\enn
which can be obtained using FFT, we can compute the tangential derivative of $\varphi$ at $x(t_j),j=0,\cdots,2N-1$ as
\ben
\frac{d\varphi}{ds}\Big|_{x=x(t_j)}=\frac{1}{|x'(t_j)|}\sum\limits_{m=-N}^{N-1}im \varphi_m e^{imt_j}
\enn
which can also be evaluated by FFT.

\section{Numerical experiments}
\label{sec:5}

In this section, various numerical examples are presented to demonstrate the accuracy and efficiency of the proposed RBIE method for solving the poroelastic problems in two dimensions. We always choose the values of parameters in (\ref{model}) as $\mu=1$, $\nu_p=0.32$, $\nu_u=0.44$, $B=0.9$, $\phi =0.25$, $\rho_f=1$, $\rho_s=3$, $\rho_a=0.15$, $\kappa=0.5$. The fully complex version of the iterative solver GMRES would be utilized to produce the solutions of the integral equations. The relative maximum errors presented in this section are calculated in accordance with the expression
\be
\varepsilon_\infty :=\frac{\max_{x\in S}\left\{ \left|U^{\rm{num}}(x)-U^{\rm{exa}}(x)\right| \right\}}{\max_{x\in S}\left\{ \left|U^{\rm{exa}}(x)\right| \right\}},
\en
where $S=\{|x|=6,x\in\R^2\}\subset \Omega^c$ and $U^{\rm{exa}}$ is the exact solution. All of the numerical tests are obtained using MATLAB.

\begin{table}[htb]
	\centering
	\caption{Numerical errors of solutions for the problem of poroelastic scattering by a kite-shaped or rounded-triangle-shaped obstacle. GMRES tol: $10^{-12}$.}
	\begin{tabular}{|c|c|c|c|c|c|}
		\hline
		\multirow{2}{*}{$\omega$} & \multirow{2}{*}{$N$} & \multicolumn{2}{|c|}{Boundary (\ref{RTS})} & \multicolumn{2}{|c|}{Boundary (\ref{KS})} \\
		\cline{3-6}
		& & CBIE & RBIE & CBIE & RBIE \\
		\hline
		&  20 &$1.45\times10^{-2}$& $6.10\times10^{-3}$ &$2.15\times10^{-3}$ &$1.43\times10^{-3}$   \\
		1 & 40 &$5.68\times10^{-6}$& $1.46\times10^{-6}$  &$1.39\times10^{-7}$ &$1.35\times10^{-7}$   \\
		& 80 & $1.28\times10^{-12}$  & $1.24\times10^{-12}$  &$9.98\times10^{-13}$ &$4.64\times10^{-13}$   \\
		\hline
		&  100 &$2.87\times10^{-2}$& $1.29\times10^{-2}$&$1.42\times10^{-2}$ &$2.10\times10^{-2}$    \\
		10 & 150  & $1.41\times10^{-3}$  & $1.06\times10^{-5}$ &$1.00\times10^{-5}$ &$5.06\times10^{-6}$   \\
		& 200 &$7.04\times10^{-9}$  & $1.70\times10^{-11}$  &$1.03\times10^{-12}$ &$1.08\times10^{-12}$   \\
		\hline
	\end{tabular}
	\label{Example1.1}
\end{table}

\begin{figure}[htb]
	\centering
	\begin{tabular}{ccc}
		\includegraphics[scale=0.15]{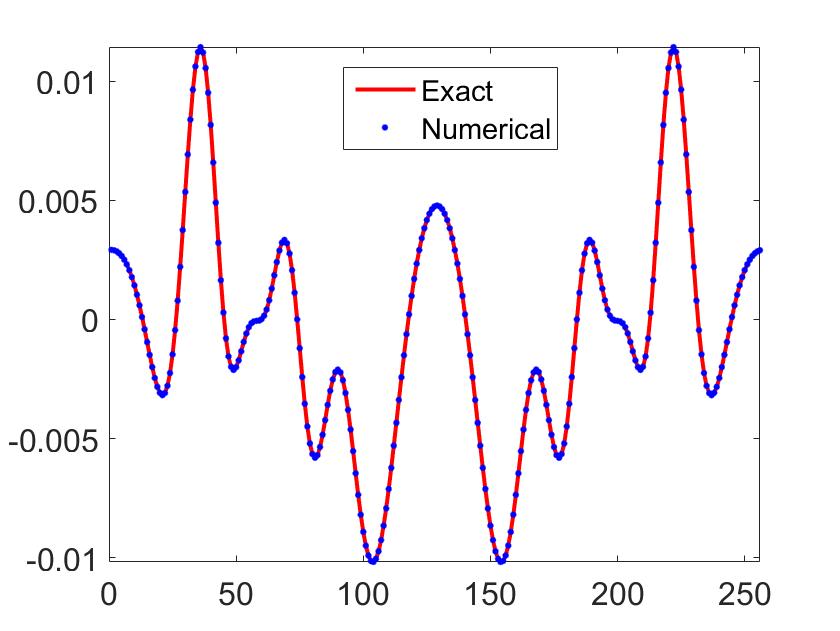} &
		\includegraphics[scale=0.15]{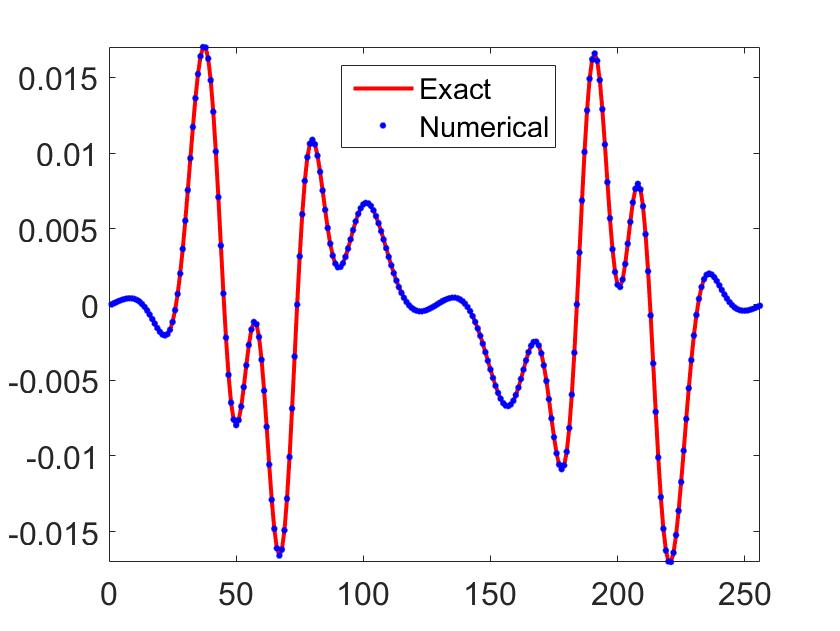} &
		\includegraphics[scale=0.15]{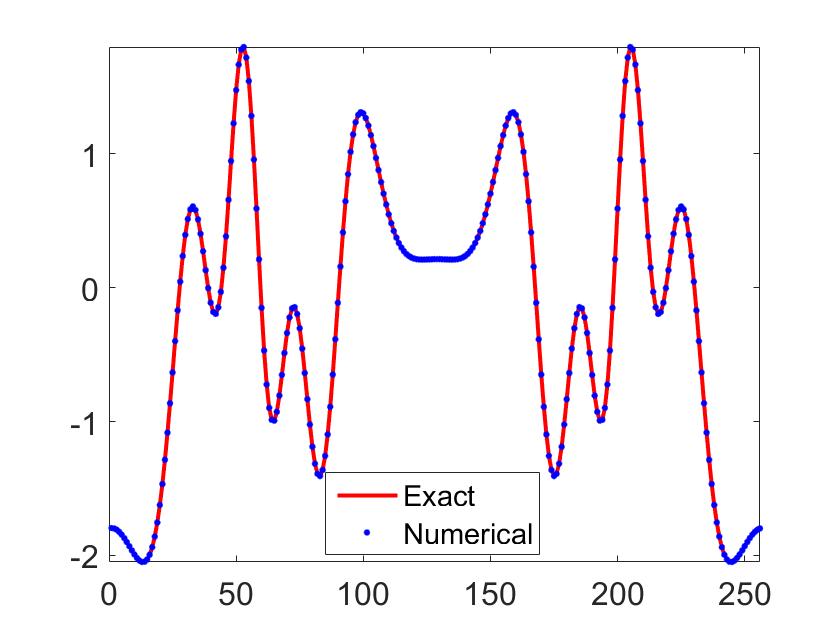} \\
		(a) $\mbox{Re}(u_1)$ & (b) $\mbox{Re}(u_2)$ & (c) $\mbox{Re}(p)$ \\
		\includegraphics[scale=0.15]{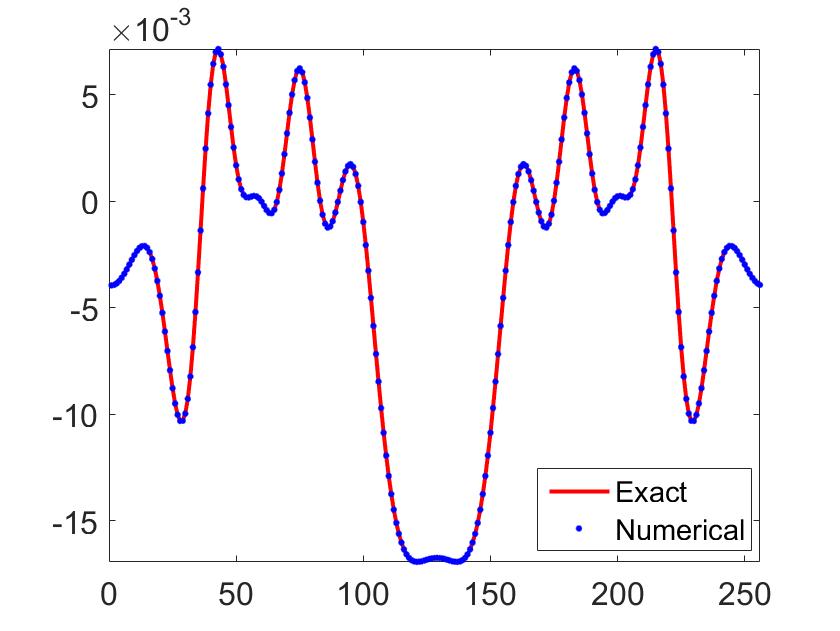} &
		\includegraphics[scale=0.15]{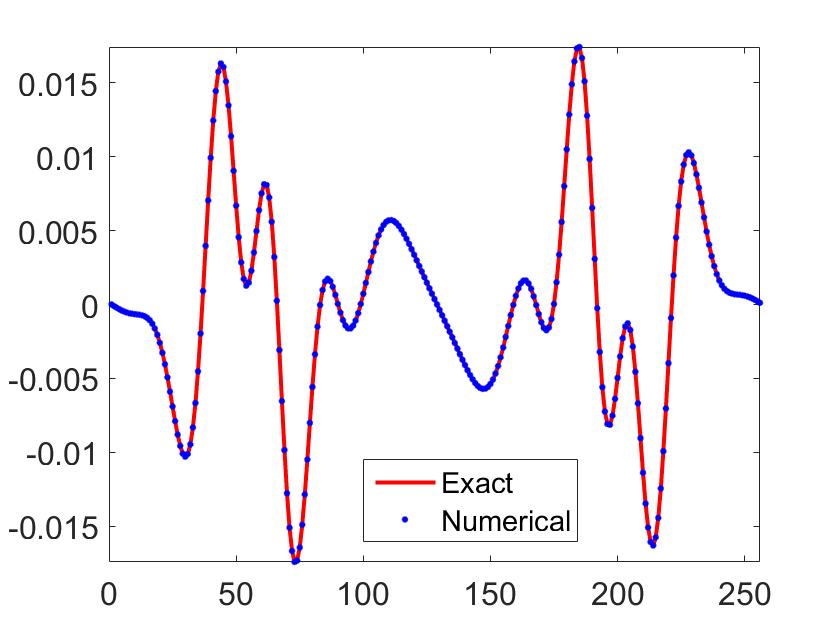} &
		\includegraphics[scale=0.15]{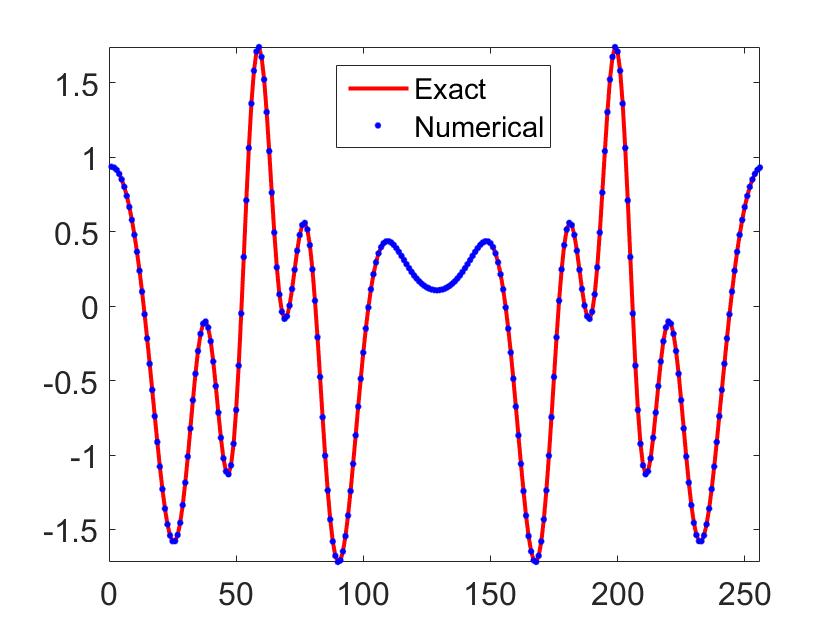} \\
		(d) $\mbox{Im}(u_1)$ & (e) $\mbox{Im}(u_2)$ & (f) $\mbox{Im}(p)$ \\
	\end{tabular}
	\caption{Comparison of the exact and numerical solutions with $N=256$ for $\omega=20$, and the relative error is $4.76\times10^{-5}$.}
	\label{Example1.2}
\end{figure}

In our first example, we consider a kite-shaped or rounded-triangle-shaped obstacle with boundary $\Gamma$ characterized by
\be
\label{RTS}
x(t)=(2+ 0.5\cos 3t)(\cos t,\sin t), \qquad t \in \left[0, 2\pi \right),
\en
and
\be
\label{KS}
x(t)=(\cos t+0.65\cos 2t-0.65,1.5\sin t), \qquad t \in \left[0, 2\pi \right),
\en
respectively, and the exact solution is given by
\begin{equation*}
u(x)=E_{12}(x,z), \qquad p(x)=E_{22}(x,z), \qquad x\in \Omega^c,
\end{equation*}
with $z=(0,0.5)^\top\in\Omega$. Table~\ref{Example1.1} displays the numerical errors of solutions using the un-regularized integral equation (\ref{BIE1}) or the regularized integral equation (\ref{BIER}) with respect to $N$, which demonstrate the high accuracy and rapid convergence of newly proposed method. The numerical and exact solutions on $S$ for the cases of kite-shaped obstacle, $N=256$, $\omega=20$ are presented in Figure~\ref{Example1.2} and the relative maximum error is $4.76\times10^{-5}$.

\begin{table}[htb]
	\centering
	\caption{Iterations required for the problem of poroelastic scattering by a kite-shaped or rounded-triangle-shaped obstacle. GMRES tol: $10^{-5}$.}
	\begin{tabular}{|c|c|c|c|c|c|}
		\hline
		\multirow{2}{*}{$\omega$} & \multirow{2}{*}{$N$} & \multicolumn{2}{|c|}{Boundary (\ref{RTS})} & \multicolumn{2}{|c|}{Boundary (\ref{KS})} \\
		\cline{3-6}
		& & CBIE & RBIE & CBIE & RBIE \\
		\hline
		1&  20 &38& 22 &48 &17   \\
		5&  100 &132& 52 &242 &39   \\
		10 & 200 &285&89 &468 &66   \\
		20& 400 & 590  &175   & 781&126   \\
		30& 600 & 893  &234   & 1005& 171  \\
		\hline
	\end{tabular}
	\label{Example2.1}
\end{table}

\begin{table}[htb]
	\centering
	\caption{Iterations required by various integral equations for the problem of poroelastic scattering by a kite-shaped obstacle. GMRES tol: $10^{-5}$.}
	\begin{tabular}{|c|c|c|c|c|}
		\hline
		$\omega$ & $N$ & Eqn.~(\ref{BIE2}) &
		RBIE with $\mathcal{R}=S$ & RBIE with $\mathcal{R}=S_0$\\
		\hline
		1 &  20 &16  &15   &17   \\
		5 & 100 &44  &36   &39   \\
		10& 200 &84  &64   &66   \\
		20& 400 &170 &133  &126  \\
		30& 600 &248 &189  &171  \\
		\hline
	\end{tabular}
	\label{Example2.2}
\end{table}

Table~\ref{Example2.1} presents the iteration numbers required to achieve the GMRES tolerance $10^{-5}$ in solving the boundary integral equations (\ref{BIE1}) and (\ref{BIER}). It can be seen that the regularized integral equation (\ref{BIER}) requires much smaller number of iterations than the un-regularized integral equation (\ref{BIE1}). It is mentioned in Remark~\ref{BIEremark} that the single-layer potential (\ref{singlelayer}) can be used to represent the solution and the integral equation (\ref{BIE2}) is obtained for the unknown potential. In addition, according to the
Caldr\'on relation studied in Section~\ref{sec:3.1}, the regularized operator $\mathcal{R}$ can also be selected as the single-layer operators $S$, and the resulted regularized boundary integral equation enjoys the same spectral property satisfied by (\ref{BIER}). Table~\ref{Example2.2} lists the iteration numbers required by various integral equations to achieve the GMRES tolerance $10^{-5}$ and it shows that the regularized integral equations (\ref{BIER}) with $\mathcal{R}=S_0$ or $\mathcal{R}=S$ require almost the same number of iterations which is smaller than that required by equation (\ref{BIE2}). In Figure~\ref{Example2.3}, we plot the determinants of the stiffness matrixes resulting from the discretization of the single-layer operator $S$ and the regularized integral equation with $\mathcal{R}=S$ or $\mathcal{R}=S_0$. It can be seen that significant troughs appear at the same frequencies in Figure~\ref{Example2.3}(a,b) which generally shows the existence of interior (Dirichlet) eigenfrequencies, however, the selection of $\mathcal{R}=S_0$ can ensures the unique solvability of (\ref{BIER}) for all frequencies from numerical point of view to a certain context.

\begin{figure}[htb]
	\centering
	\begin{tabular}{ccc}
		\includegraphics[scale=0.15]{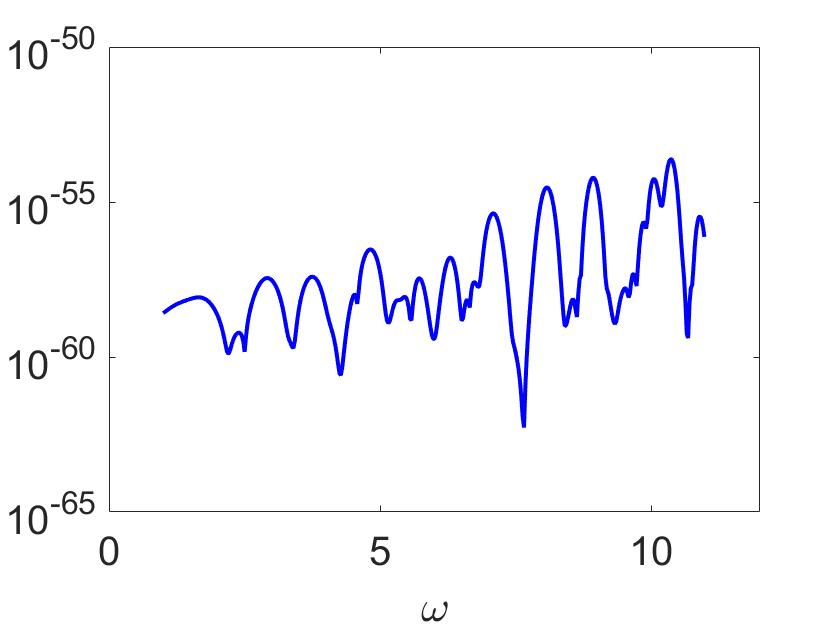} &
		\includegraphics[scale=0.15]{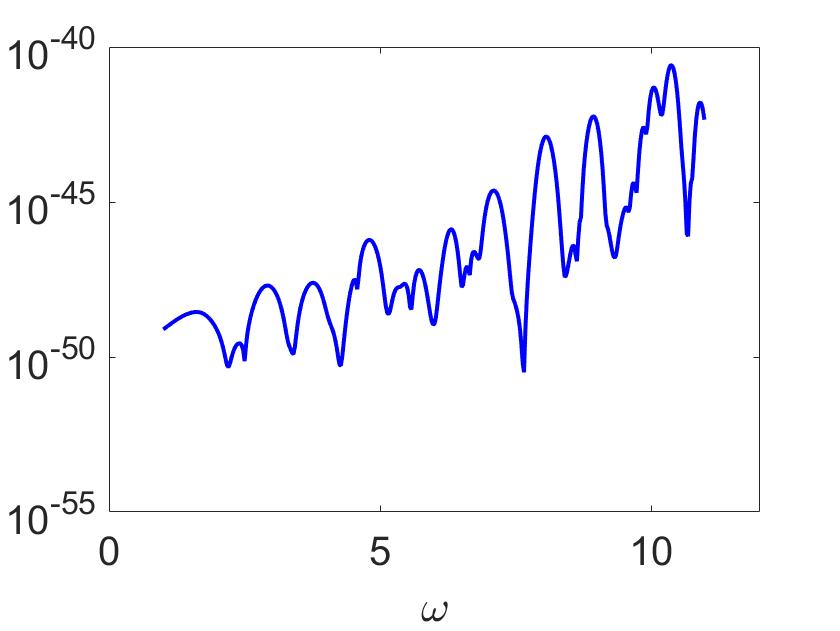} &
		\includegraphics[scale=0.15]{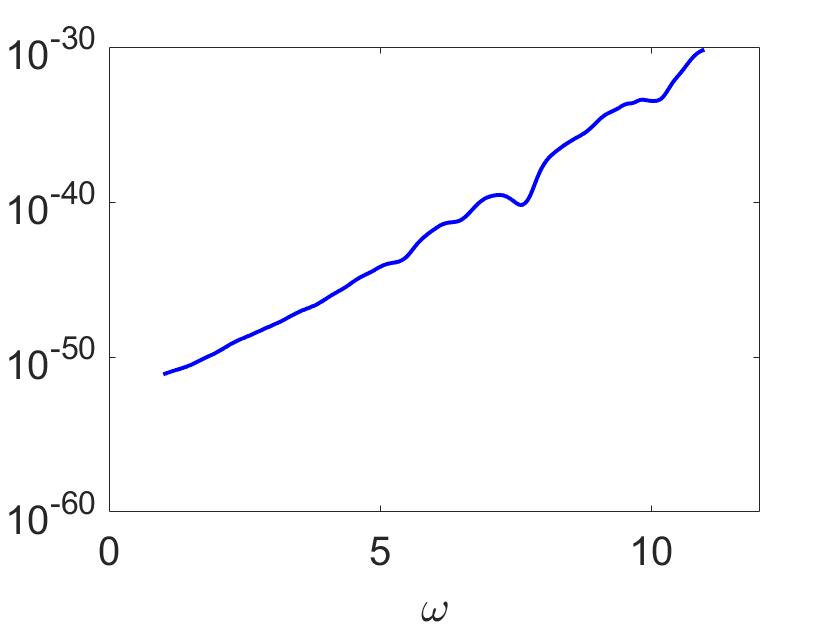} \\
		(a) $S$ & (b) Eqn.~(\ref{BIER}) with $\mathcal{R}=S$ & (c) Eqn.~(\ref{BIER}) with $\mathcal{R}=S_0$  \\
	\end{tabular}
	\caption{The determinant of the stiffness matrix with respect to frequency.}
	\label{Example2.3}
\end{figure}

Finally, we consider the scattering of an incident point source $U^{inc}=({u^{inc}}^\top,p^{inc})^\top$ in the form
\begin{equation*}
u^{inc}(x)=E_{12}(x,z),\qquad p^{inc}(x)=E_{22}(x,z)
\end{equation*}
by a kite-shaped obstacle where $z$ denotes the location of the point source. The numerical solutions in $\Omega^c$ with $\omega=20$ are presented in Figures~\ref{Example3.1} for $z=(3,3)$ and $z=(-3,0)$, respectively.

\begin{figure}[htb]
	\centering
	\begin{tabular}{cc}
		\includegraphics[scale=0.3]{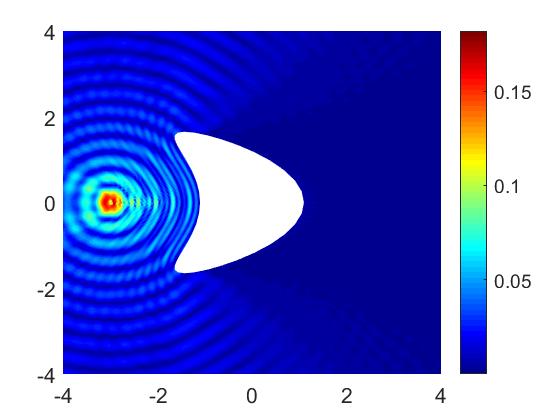} &
		\includegraphics[scale=0.3]{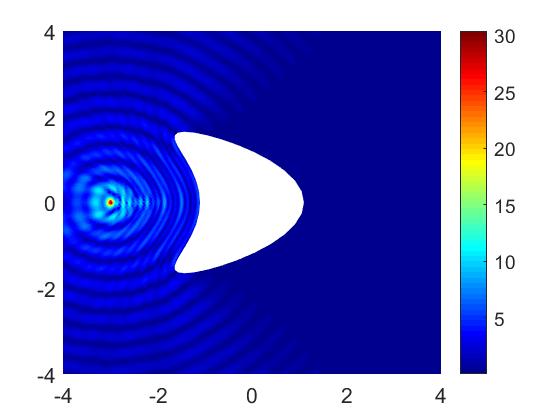} \\
		(a) $|u|$ & (b) $|p|$ \\
		\includegraphics[scale=0.3]{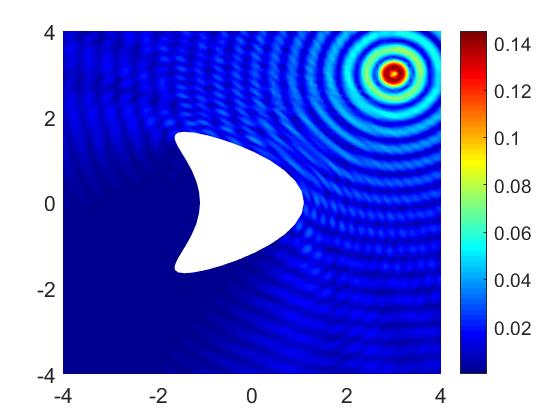} &
		\includegraphics[scale=0.3]{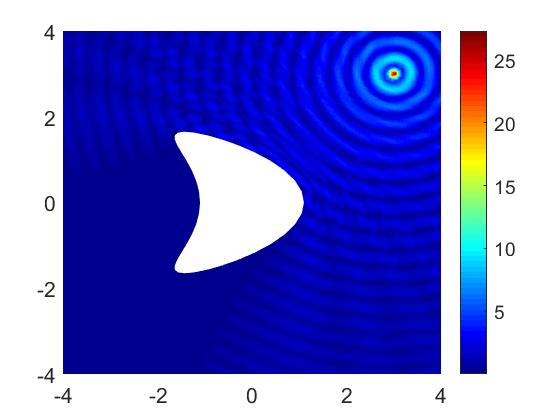} \\
		(d) $|u|$ & (e) $|p|$ \\
	\end{tabular}
	\caption{The total field $U$ for the scattering of an incident point source by a kite-shaped obstacle. (a,b): $z=(-3,0)$; (c,d): $z=(3,3)$.}
	\label{Example3.1}
\end{figure}

\section{Conclusion}
\label{sec:6}

We propose in this work a novel regularized integral equation method for solving the exterior Neumann boundary value problem of poroelastic wave scattering in two dimensions. On a basis of the spectral properties of poroelastic double-layer integral operator and the corresponding Caldr\'on relation, and the new regularized formulations of strongly-singular and hyper-singular operators, the proposed approach possesses spectral convergence (using Nystr\"om discretization method) and reduces the number of GMRES iterations consistently across various geometries and frequency regimes. Theoretical analysis of the unique solvability of combined integral equation (\ref{BIE1}), the combination of the regularized integral equation method with other popular fast solvers such as fast multipole method, the application of the regularized integral equation method for three dimensional problems and solving problems with non-smooth surfaces are left for future work.

\section*{Acknowledgement}
LWX is partially supported by Key Project of the Major Research Plan of NSFC (91630205) and a grant of NSFC (11771068 ).

\appendix
\section*{Appendix A. Proofs of Lemmas~\ref{regKP}-\ref{regN4}}
\renewcommand{\theequation}{A.\arabic{equation}}

This appendix presents the main idea for the proofs of Lemmas~\ref{regKP}-\ref{N4}. Recall that the traction operator can be rewritten as
\be
T(\partial ,\nu )u(x) = (\lambda  + 2\mu )\nu (\nabla  \cdot u) + \mu {\partial _\nu }u + \mu M(\partial ,\nu )u.
\label{eq.tra}
\en
It follows that
\be
{\partial _\nu }\nabla  - M(\partial ,\nu )\nabla  = \nu \Delta,
\label{eq.G1}
\en
and
\be
T(\partial ,\nu )\nabla  = (\lambda  + 2\mu )\nu \Delta  + 2\mu M(\partial ,\nu )\nabla.
\label{eq.G2}
\en

{\it Proof of Lemma~\ref{regKP}.}
We first investigate the operator $K'$. It is known from~\cite{BXY191} that for $x \ne y$,
\be
T(\partial_x,\nu_x)E_{11}(x,y) &=& - \nu _x\nabla _x^\top\left[ \gamma_{k_s}(x,y) - \gamma _{k_1}(x,y) \right] + \frac{k_2^2 - q}{k_1^2 - k_2^2}\nu _x\nabla _x^\top\left[ \gamma _{k_1}(x,y) - \gamma _{k_2}(x,y) \right]\\ \nonumber
&\quad& + \partial _{\nu _x}\gamma _{k_s}(x,y)I + M(\partial_x ,\nu_x)\left[ 2\mu E_{11}(x,y) - \gamma _{k_s}(x,y)I \right],
\label{eq.TEx}
\en
and
\be
T(\partial_y,\nu_y)E_{11}(x,y) &=& - \nu _y\nabla _y^\top\left[ \gamma_{k_s}(x,y) - \gamma _{k_1}(x,y) \right] + \frac{k_2^2 - q}{k_1^2 - k_2^2}\nu _y\nabla _y^\top\left[ \gamma _{k_1}(x,y) - \gamma _{k_2}(x,y) \right]\\ \nonumber
&\quad& + \partial _{\nu _y}\gamma _{k_s}(x,y)I + M(\partial_y ,\nu_y)\left[ 2\mu E_{11}(x,y) - \gamma _{k_s}(x,y)I \right].
\label{eq.TEy}
\en
Then it can be obtained that
\ben
K'_1(u)(x) &=& \int_\Gamma \nu _x\nabla _x^\top \left[- (\gamma _{k_s}(x,y) - \gamma _{k_1}(x,y)) + \frac{k_2^2 - q}{k_1^2 - k_2^2}(\gamma _{k_s}(x,y) - \gamma _{k_1}(x,y))\right]u(y)ds_y\\
&\quad& + \int_\Gamma  \left[\partial _{\nu _x}\gamma _{k_s}(x,y)I - \alpha \nu _xE_{12}^\top(x,y) +\frac{d}{ds_x}A(2\mu E_{11}(x,y) - \gamma _{k_s}(x,y)I) \right]u(y)ds_y.
\enn
It follows from (\ref{eq.G2}) that
\ben
T(\partial_x,\nu_x)E_{21}(x,y) &=& -\frac{\gamma}{(\lambda+2\mu)(k_1^2-k_2^2)}T(\partial_x,\nu_x)\nabla_x \left[\gamma_{k_1}(x,y)-\gamma_{k_2}(x,y)\right]\\
&=&-\frac{\gamma}{(k_1^2-k_2^2)}\nu_x \Delta_x \left[\gamma_{k_1}(x,y)-\gamma_{k_2}(x,y)\right]\\
&\quad& -\frac{2\mu\gamma}{(\lambda+2\mu)(k_1^2-k_2^2)} M(\partial_x ,\nu_x)\nabla_x\left[\gamma_{k_1}(x,y)-\gamma_{k_2}(x,y)\right]\\
&=&\frac{\gamma}{(k_1^2-k_2^2)}\left[k_1^2\gamma_{k_1}(x,y) -k_2^2\gamma_{k_2}(x,y)\right]\nu_x\\
&\quad& -\frac{2\mu\gamma}{(\lambda+2\mu)(k_1^2-k_2^2)} A\frac{d}{ds_x}\nabla_x\left[\gamma_{k_1}(x,y)-\gamma_{k_2}(x,y)\right].
\enn
Therefore,
\ben
K'_2(p)(x) &=&\int_\Gamma  \left[\frac{\gamma}{k_1^2 - k_2^2} (k_1^2\gamma _{k_1}(x,y)-k_2^2\gamma _{k_2}(x,y)) - \alpha \nu _xE_{22}(x,y)\right]p(y)ds_y\\
&\quad&-\int_\Gamma\frac{2\mu\gamma}{(\lambda+2\mu)(k_1^2-k_2^2)}A\frac{d}{ds_x}\nabla_x\left[\gamma_{k_1}(x,y)-\gamma_{k_2}(x,y)\right]p(y)ds_y.
\enn
On the other hand, we obtain from (\ref{eq.G1}) that
\ben
\partial_{\nu_x}E_{12}^\top(x,y) &=& \frac{i\omega \gamma}{(\lambda+2\mu)(k_1^2-k_2^2)}\partial_{\nu_x}\nabla_x^\top\left[\gamma_{k_1}(x,y)-\gamma_{k_2}(x,y)\right]\\
&=&\frac{i\omega \gamma}{(\lambda+2\mu)(k_1^2-k_2^2)}\nu_x \Delta_x\left[\gamma_{k_1}(x,y)-\gamma_{k_2}(x,y)\right]\\
&\quad&+\frac{i\omega \gamma}{(\lambda+2\mu)(k_1^2-k_2^2)}\left\{M(\partial_x,\nu_x)\nabla_x \left[\gamma_{k_1}(x,y)-\gamma_{k_2}(x,y)\right] \right\}^\top\\
&=&-\frac{i\omega\gamma}{(\lambda+2\mu)(k_1^2-k_2^2)}\left[k_1^2\gamma_{k_1}(x,y)-k_2^2\gamma_{k_2}(x,y)\right]\nu_x\\
&\quad&-\frac{i\omega\gamma}{(\lambda+2\mu)(k_1^2-k_2^2)}\frac{d}{ds_x}\nabla_x^\top\left[\gamma_{k_1}(x,y)-\gamma_{k_2}(x,y)\right]A.
\enn
Then we have
\ben
K^{\prime}_3(u)(x) &=& \int_\Gamma\left[-i\omega\beta\nu_x^\top E_{11}(x,y) + C_0(k_1^2\gamma _{k_1}(x,y) - k_2^2\gamma _{k_2}(x,y))\nu _x^\top\right] u(y)ds_y\\
&\quad& + \int_\Gamma C_0\frac{d}{ds_x}\nabla _x^\top\left[\gamma _{k_1}(x,y) - \gamma_{k_2}(x,y)\right]Ap(y)ds_y.
\enn
The formula for $K_4'$ can be obtained directly from its definition and this completes the proof of Lemma~\ref{regKP}.

{\it Proof of Lemmas~\ref{regN1}-\ref{regN4}.}
Note that the hyper-singular integral operator of $N$ is given by
\ben
N(\psi)(x)=\begin{bmatrix}
	N_1 & N_2\\
	N_3 & N_4
\end{bmatrix}\begin{bmatrix}
	u\\
	p
\end{bmatrix}(x),\quad x\in\Gamma,\quad \psi=(u^\top,p)^\top,
\enn
where
\ben
N_1(u)(x) &=& \int_\Gamma\left[T(\partial _x,\nu _x)(T(\partial _y,\nu _y)E_{11}(x,y))^\top - i\omega \alpha T(\partial _x,\nu _x)E_{21}(x,y)\nu _y^\top \right]u(y)ds_y\\
&\quad&+\int_\Gamma\left[-\alpha\nu _x(T(\partial _y,\nu _y)E_{12}(x,y))^\top + i\omega\alpha^2E_{22}(x,y)\nu _x\nu _y^\top \right]u(y)ds_y\\
N_2(p)(x) &=& \int_\Gamma\left[-\beta T(\partial _x,\nu _x)E_{11}(x,y)\nu _y + \frac{i\beta}{\rho_f\omega}\partial_{\nu _y}T(\partial _x,\nu _x)E_{21}(x,y) \right]p(y)ds_y\\
&\quad&+\int_\Gamma\left[\alpha\beta\nu_xE_{12}^\top\nu _y-\frac{i\alpha \beta}{\rho_f\omega}\partial_{\nu_y}E_{22}(x,y)\nu_x \right]p(y)ds_y\\
N_3(u)(x) &=& -\int_\Gamma\left[i\omega\beta\nu_x^\top(T(\partial _y,\nu_y)E_{11}(x,y))^\top+\omega^2\alpha\beta\nu_x^\top E_{21}\nu_y^\top \right]u(y)ds_y\\
&\quad&+\int_\Gamma\left[\frac{i\beta}{\rho_f\omega}\partial_{\nu_x}(T(\partial_y,\nu_y)E_{12}(x,y))^\top+\frac{\alpha\beta}{\rho_f}\partial_{\nu_x}E_{22}(x,y)\nu_y^\top\right]u(y)ds_y\\
N_4(p)(x) &=& \int_\Gamma\left[i\omega\beta^2\nu _x^\top E_{11}(x,y)\nu_y + \frac{\beta^2}{\rho_f}\nu _x^\top\partial_{\nu _y}E_{21}(x,y) \right]p(y)ds_y\\
&\quad&-\int_\Gamma\left[\frac{i\beta^2}{\rho_f\omega}\partial_{\nu_x}E_{12}^\top(x,y)\nu_y+\frac{\beta^2}{\rho_f^2\omega^2}\partial_{\nu_x}(\partial_{\nu_y}E_{22}(x,y)) \right]p(y)ds_y.
\enn
We first consider the following term
\ben
T(\partial_x,\nu_x)\int_\Gamma(T(\partial_y,\nu_y)E_{11}(x,y))^\top u(y)ds_y.
\enn
Set
\ben
f_1(x)&=&\int_\Gamma\nabla_y\left[\gamma_{k_s}(x,y)-\gamma_{k_1}(x,y)\right]\nu_y^\top u(y)ds_y,\\
f_2(x)&=&\int_\Gamma\nabla_y\left[\gamma_{k_1}(x,y)-\gamma_{k_2}(x,y)\right]\nu_y^\top u(y)ds_y,\\
f_3(x)&=&\int_\Gamma\partial_{\nu_y}\gamma_{k_s}(x,y) u(y)ds_y,\\
f_4(x)&=&\int_\Gamma\left[2\mu E_{11}(x,y)-\gamma_{k_s}(x,y)I\right]M(\partial_y,\nu_y)u(y)ds_y,
\enn
and
\ben
g_i(x)=\mu\partial_{\nu_x}f_i(x) + (\lambda + \mu )\nu_x\nabla_x\cdot f_i(x) + \mu M(\pa_x,\nu_x)f_i(x).
\enn
Then we obtain from (\ref{eq.tra}) that
\be
g_1(x)&=&(\lambda+2\mu)\int_\Gamma\left[k_s^2\gamma_{k_s}(x,y)-k_1^2\gamma_{k_1}(x,y)\right]\nu_x\nu _y^\top u(y)d{s_y}\nonumber\\
&\quad& + 2\mu \frac{d}{ds_x}\int_\Gamma A\nabla_y\left[\gamma_{k_s}(x,y) - \gamma _{k_1}(x,y)\right]\nu _y^\top u(y)ds_y,
\label{eq.g1}
\en
and
\be
g_2(x)&=&(\lambda+2\mu)\int_\Gamma\left[k_1^2\gamma_{k_1}(x,y)-k_2^2\gamma_{k_2}(x,y)\right]\nu_x\nu _y^\top u(y)ds_y\nonumber\\
&\quad&+2\mu\frac{d}{ds_x}\int_\Gamma A\nabla_y\left[\gamma _{k_1}(x,y)-\gamma_{k_2}(x,y)\right]\nu _y^\top u(y)ds_y.
\label{eq.g2}
\en
It follows from \cite{BXY191} that $g_3(x)$ can be expressed as
\be
g_3(x)&=&\mu\frac{d}{ds_x}\int_\Gamma\gamma_{k_s}(x,y)\frac{du}{ds_y}ds_y + \mu k_s^2\int_\Gamma\gamma_{k_s}(x,y)\nu_x^\top\nu_yu(y)ds_y\nonumber\\
&\quad&+(\lambda+\mu )\int_\Gamma\nu_x\nabla_x^\top\partial_{\nu_y}\gamma _{k_s}(x,y)u(y)ds_y + \mu\frac{d}{ds_x}\int_\Gamma\partial_{\nu _y}\gamma _{k_s}(x,y)u(y)ds_y.
\label{eq.g3}
\en
For $g_4(x)$, we know from (\ref{eq.TEx}) that
\be
g_4(x) &=& \mu\int_\Gamma\nu_x^\top\nabla_x\gamma_{k_s}(x,y) A\frac{du}{ds_y}ds_y - 2\mu\int_\Gamma\nu_x\nabla_x^\top\left[\gamma_{k_s}(x,y) - \gamma_{k_1}(x,y)\right] A\frac{du}{ds_y}ds_y\nonumber\\
&\quad& + \frac{2\mu(k_1^2-q)}{k_1^2-k_2^2}\int_\Gamma\nu_x\nabla_x^T\left[\gamma _{k_1}(x,y) - \gamma_{k_2}(x,y)\right] A\frac{du}{ds_y}ds_y\nonumber\\
&\quad& + 4\mu^2\frac{d}{ds_x}\int_\Gamma AE_{11}(x,y)A \frac{du}{ds_y}ds_y + 3\mu\frac{d}{ds_x}\int_\Gamma\gamma_{k_s}(x,y)\frac{du}{ds_y}ds_y\nonumber\\
&\quad& - (\lambda  + \mu )\int_\Gamma\nu_x\nabla_x^\top\gamma_{k_s}(x,y)A \frac{du}{ds_y}ds_y.
\label{eq.g4}
\en
Therefore, (\ref{eq.g1})-(\ref{eq.g4}) yields
\be
&\quad&T(\partial_x,\nu_x)\int_\Gamma(T(\partial_y,\nu_y)E_{11}(x,y))^\top u(y)ds_y\\ \nonumber
&=& -g_1(x)+\frac{k_2^2-q}{k_1^2-k_2^2}g_2(x)+g_3(x)+g_4(x)\\ \nonumber
&=&  -(\rho-\beta\rho_f)\omega^2\int_\Gamma\gamma_{k_s}(x,y)(\nu_x\nu_y^\top - \nu_x^\top\nu_yI - J_{\nu_x,\nu_y})u(y)ds_y\\ \nonumber
&\quad&+\int_\Gamma\left[\frac{k_1^2(\lambda + 2\mu )(k_1^2 - q)}{k_1^2 - k_2^2}\gamma_{k_1}(x,y) - \frac{k_2^2(\lambda  + 2\mu )(k_2^2 - q)}{k_1^2 - k_2^2}\gamma_{k_2}(x,y)\right] \nu_x\nu _y^\top u(y)ds_y\\ \nonumber
&\quad&+ 4\mu\frac{d}{ds_x}\int_\Gamma\gamma_{k_s}(x,y) \frac{du}{ds_y}ds_y + 4\mu^2\frac{d}{ds_x}\int_\Gamma AE_{11}(x,y)A\frac{du}{ds_y}ds_y\\ \nonumber
&\quad&-2\mu\int_\Gamma\nu_x\nabla_x^\top\left[\gamma_{k_s}(x,y) - \gamma _{k_1}(x,y)\right] A\frac{du}{ds_y}ds_y\\ \nonumber
&\quad& - 2\mu \frac{d}{ds_x}\int_\Gamma A\nabla_y\left[\gamma_{k_s}(x,y) - \gamma _{k_1}(x,y)\right] \nu _y^\top u(y)ds_y\\ \nonumber
&\quad&+\frac{2\mu(k_2^2 - q)}{k_1^2 - k_2^2}\int_\Gamma \nu_x\nabla _x^\top\left[\gamma_{k_1}(x,y) - \gamma_{k_2}(x,y)\right] A\frac{du}{ds_y}ds_y\\ \nonumber
&\quad& + \frac{2\mu(k_2^2 - q)}{k_1^2 - k_2^2}\frac{d}{ds_x}\int_\Gamma  A\nabla _y\left[\gamma_{k_1}(x,y) - \gamma_{k_2}(x,y)\right]\nu_y^\top u(y)ds_y.
\label{eq.N11}
\en
On the other hand, we obtain from (\ref{eq.G2}) that
\be
&\quad&\int_\Gamma T(\partial_x,\nu_x)E_{21}(x,y)\nu_y^\top u(y)ds_y\\ \nonumber
&=& -\frac{\gamma}{k_1^2 - k_2^2}\int_\Gamma\left[k_1^2\gamma_{k_1}(x,y)- k_2^2\gamma_{k_2}(x,y)\right] \nu_x\nu_y^\top u(y)ds_y\\ \nonumber
&\quad& -\frac{2\mu\gamma}{(\lambda  + 2\mu )(k_1^2 - k_2^2)}\frac{d}{ds_x}\int_\Gamma A\nabla_y\left[\gamma_{k_1}(x,y) - \gamma _{k_2}(x,y)\right]\nu_y^\top u(y)ds_y,
\label{eq.N12}
\en
and
\be
&\quad&\int_\Gamma\nu _x(T(\partial_y,\nu_y)E_{12}(x,y))^\top u(y)ds_y\\ \nonumber
&=& -\frac{\gamma}{k_1^2 - k_2^2}\int_\Gamma\left[k_1^2\gamma_{k_1}(x,y) - k_2^2\gamma_{k_2}(x,y)\right]\nu_x\nu_y^\top u(y)ds_y\\ \nonumber
&\quad& -\frac{2\mu\gamma}{(\lambda  + 2\mu )(k_1^2 - k_2^2)}\frac{d}{ds_x}\int_\Gamma\nu_x\nabla_x^\top(\gamma_{k_1}(x,y) - \gamma _{k_2}(x,y)) A\frac{du}{ds_y}ds_y.
\label{eq.N13}
\en
Then Lemma~\ref{regN1} can be proved by combining (\ref{eq.N11})-(\ref{eq.N13}). The proofs of Lemma~\ref{regN2}-\ref{regN4} are analogous, and thus are omitted here.

\end{document}